\documentclass[11pt]{article}%
\usepackage{amssymb}
\usepackage{amsmath}
\usepackage[onehalfspacing]{setspace}
\usepackage{geometry}
\usepackage{amsfonts}
\usepackage{sectsty}
\usepackage{chicago}
\usepackage{graphicx}%
\providecommand{\U}[1]{\protect\rule{.1in}{.1in}}
\newtheorem{theorem}{Theorem}

\newtheorem{proposition}[theorem]{Proposition}

\newenvironment{proof}[1][Proof]{\noindent\textbf{#1.} }{\ \rule{0.5em}{0.5em}}
\begin{document}
%
%TCIMACRO{\TeXButton{Section}{\sectionfont{\bfseries\large\sffamily}}}%
%BeginExpansion
\sectionfont{\bfseries\large\sffamily}%
%EndExpansion
%

%TCIMACRO{\TeXButton{Subsection}{\subsectionfont{\bfseries\sffamily\normalsize
%}}}%
%BeginExpansion
\subsectionfont{\bfseries\sffamily\normalsize}%
%EndExpansion
%

%TCIMACRO{\TeXButton{noindent}{\noindent}}%
%BeginExpansion
\noindent
%EndExpansion%
%TCIMACRO{\TeXButton{title}{{\sffamily\bfseries\Large
%Cross-screening in observational studies that test many hypotheses}}}%
%BeginExpansion
{\sffamily\bfseries\Large
Cross-screening in observational studies that test many hypotheses}%
%EndExpansion
%

%TCIMACRO{\TeXButton{noindent}{\noindent}}%
%BeginExpansion
\noindent
%EndExpansion
\textsf{Qingyuan Zhao, Dylan S. Small, Paul R. Rosenbaum}%
\footnote{\textit{Address for correspondence:} Department of Statistics, The
Wharton School, University of Pennsylvania, Jon M. Huntsman Hall, 3730 Walnut
Street, Philadelphia, PA 19104-6340 USA. \ \textsf{E-mail:}
dsmall@wharton.upenn.edu. \ 25 February 2017.}%

%TCIMACRO{\TeXButton{noindent}{\noindent}}%
%BeginExpansion
\noindent
%EndExpansion
\textsf{University of Pennsylvania, Philadelphia}%

%TCIMACRO{\TeXButton{noindent}{\noindent}}%
%BeginExpansion
\noindent
%EndExpansion
\textsf{Abstract. \ We discuss observational studies that test many causal
hypotheses, either hypotheses about many outcomes or many treatments. \ To be
credible an observational study that tests many causal hypotheses must
demonstrate that its conclusions are neither artifacts of multiple testing nor
of small biases from nonrandom treatment assignment. \ In a sense that needs
to be defined carefully, hidden within a sensitivity analysis for nonrandom
assignment is an enormous correction for multiple testing: in the absence of
bias, it is extremely improbable that multiple testing alone would create an
association insensitive to moderate biases. \ We propose a new strategy called
\textquotedblleft cross-screening,\textquotedblright\ different from but
motivated by recent work of Bogomolov and Heller on replicability.
\ Cross-screening splits the data in half at random, uses the first half to
plan a study carried out on the second half, then uses the second half to plan
a study carried out on the first half, and reports the more favorable
conclusions of the two studies correcting using the Bonferroni inequality for
having done two studies. \ If the two studies happen to concur, then they
achieve Bogomolov-Heller replicability; however, importantly, replicability is
not required for strong control of the family-wise error rate, and either
study alone suffices for firm conclusions. In randomized studies with a few
hypotheses, cross-split screening is not an attractive method when compared
with conventional methods of multiplicity control, but it can become
attractive when hundreds or thousands of hypotheses are subjected to
sensitivity analyses in an observational study. \ We illustrate the technique
by comparing 46 biomarkers in individuals who consume large quantities of fish
versus little or no fish.}%

%TCIMACRO{\TeXButton{noindent}{\noindent}}%
%BeginExpansion
\noindent
%EndExpansion

\newpage

\section{Introduction: testing many hypotheses in observational studies}

\label{secIntro}

\subsection{Sensitivity analyses and corrections for multiple testing affect
one another}

\label{ssIntroIdea}

To be credible, an observational or nonrandomized study that tests many null
hypotheses about treatment effects must demonstrate that claimed effects are
neither an artifact of testing many hypotheses nor a consequence of small
departures from randomized treatment assignment. \ These two demonstrations
are related. \ Tiny treatment effects have almost no chance of being judged
insensitive to small unmeasured biases from nonrandom assignment; see
Rosenbaum (2015, \S 5). \ With many null hypotheses, tiny treatment effects
also have almost no chance of being distinguished from artifacts of multiple
testing in studies of moderate sample size. \ These two considerations
motivate screening to eliminate null or negligible treatment effects before
testing, thereby greatly reducing the correction for multiple testing.

To see that tiny effects have almost no chance of being distinguished from
artifacts of multiple testing in studies of moderate sample size, suppose
there are $I$ independent observations on each of $K$ independent outcomes,
where outcome $k$ is $N\left(  \tau_{k},1\right)  $, with the $k$th null
hypothesis asserting $H_{k}:\tau_{k}=0$, where in fact $\tau_{1}=\tau>0$ and
$\tau_{2}=\tau_{3}=\cdots=\tau_{K}=0$. \ In this case, the probability that
the only active group, $k=1$, has a sample mean above the $K-1$ other group
means is $\varpi_{\tau,K,I}=\int_{-\infty}^{\infty}\Phi\left(  y/\sqrt
{I}\right)  ^{K-1}\,\phi\left\{  \left(  y-\tau\right)  /\sqrt{I}\right\}
/\sqrt{I}\,dy$, where $\Phi\left(  \cdot\right)  $ and $\phi\left(
\cdot\right)  $ are, respectively, the standard Normal cumulative and density
functions. \ With $I=100$ observations and $K=100$ hypotheses, $\varpi
_{\tau,K,I}$ is 0.082 for $\tau=0.1$, 0.501 for $\tau=0.25$ and 0.988 for
$\tau=0.5$. \ So, we are unlikely to locate the one real effect if $\tau=0.1$
and very likely to locate it if $\tau=0.5$. \ As $I\rightarrow\infty$ with
$\tau$ and $K$ fixed, $\varpi_{\tau,K,I}\rightarrow1$; however, this limiting
calculation is less relevant in the problems we often face in which $K$ is
fairly large and $I$ is not enormous relative to $K$. \ An effect of
$\tau=0.5$ with $I=100$ is also likely to be judged insensitive to substantial
biases from nonrandom treatment assignment; see Rosenbaum (2010, \S 14.2, p.
268). \ As will be seen in later sections, we can reduce the magnitude of a
correction for multiple testing by screening out negligible effects that never
had a chance of surviving a sensitivity analysis, and had little chance of
surviving a correction for multiple testing.

In the current paper, we propose a new technique called cross-screening. \ The
sample is split in half at random. \ The first half is used to plan the
analysis of the second half, for instance, selecting null hypotheses that
appear to be false, the one-tailed direction of the likely departure from the
null, and the best test statistic to use in testing. \ In parallel, the second
half is used to plan the analysis of the first half. \ Both halves are
analyzed using these data-derived plans, and the more favorable results of the
two analyses are reported with a Bonferroni correction for performing two
analyses. \ Both halves are used to plan and both halves are used to test, so
cross-screening uses all of the data to test hypotheses. \ Cross-screening
strongly controls the family-wise error rate, as discussed in
\S \ref{ssCrossScreenDescribe}. \ We compare cross-screening to alternative
methods, such as a Bonferroni correction for testing $K$ two-sided hypotheses,
or screening using a small fraction of the data. \ Cross-screening performs
poorly except when $K$ is large and $I$ is not extremely large, but it often
wins decisively in sensitivity analyses with large $K$ and moderate $I/K$.
\ Cross-screening is related to, though distinct from, a concept of
replicability developed by Bogomolov and Heller (2013); see
\S \ref{ssCrossScreenReplicability} and \S \ref{secNonrandomCS} for detailed
discussion. \ Cross-screening rejects hypothesis $H_{k}$ if either half-sample
rejects $H_{k}$, thereby strongly controlling the family-wise rate; however,
it achieves Bogomolov-Heller replicability if both halves reject $H_{k}$.

\subsection{Outline: a general method; measures of performance}

As it turns out, our simplest results are also our most definite and most
useful results, so we present them first \S \ref{secCrossScreen}, after fixing
ideas with a motivating example in \S \ref{ssIntroExample}. \ Section
\ref{secEvaluation} presents a series of incomplete evaluations of
cross-screening; however, the simulation in \S \ref{ssSimulation} can be read
immediately after \S \ref{secCrossScreen}. \ Alas, the attraction of
cross-screening is its flexibility, but it is this very flexibility that makes
it difficult to offer a definitive evaluation. \ Suffice it to say here that
cross-screening performs well in certain situations and from certain vantage
points, but it performs poorly in others, and a user of the method needs to be
aware of both aspects. \ The formal evaluation in \S \ref{secEvaluation}
requires some notation and definitions that are reviewed in
\S \ref{secNotationReview}, whereas some of these same ideas appear informally
in \S \ref{secCrossScreen}. \ Finally, \S \ref{secNonrandomCS} considers
nonrandom cross-screening.

\subsection{Example: eating fish and biomarkers}

\label{ssIntroExample}

Using the National Health and Nutrition Examination Survey (NHANES) 2013-2014,
we defined high fish consumption as more than 12 servings of fish or shellfish
in the previous month, and low fish consumption as 0 or 1 servings of fish.
\ We matched 234 adults with high fish consumption to 234 adults with low fish
consumption, matching for gender, age, income, race, education and smoking.
\ Income was measured as the ratio of income to the poverty level capped at
five times poverty. \ Smoking was measured by two variables: (i) having smoked
a total of at least 100 cigarettes, and (ii) cigarettes currently smoked per
month. \ Figure 1 shows the distribution of age, income, education and smoking
for the matched high and low groups. \ For brevity, we refer to high fish
consumption as treated, and to low fish consumption as control.

\begin{figure}[t]
  \centering
  \includegraphics[width = \textwidth]{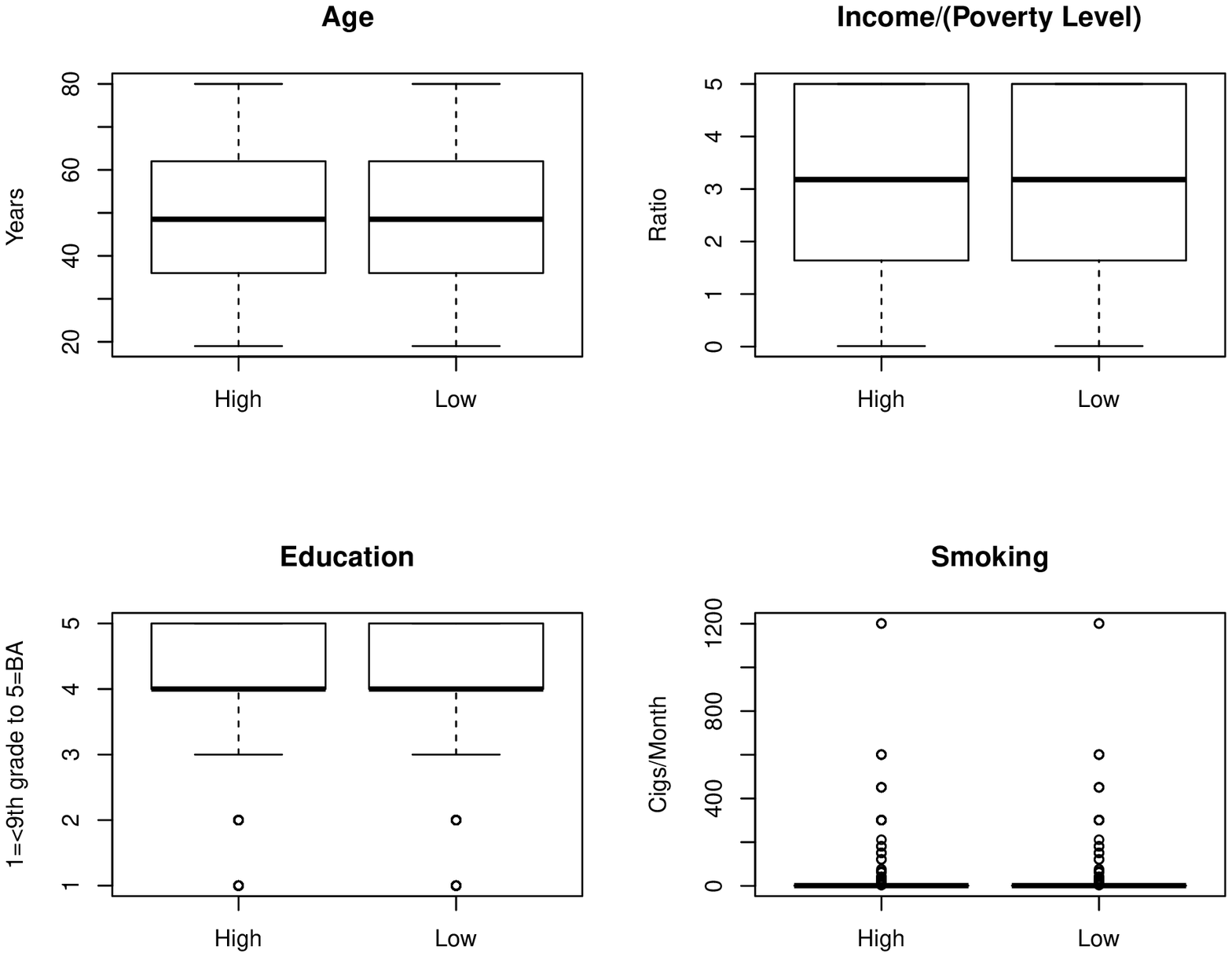}
  \caption{Balance of four covariates in $I = 234$ matched pairs of two people, one with a high consump-
    tion of fish, the other with a low consumption of fish in the past month.}
\label{fig:1}
\end{figure}

We compared the treated and control groups in terms of the logs of $K=46$
biomarkers. \ All biomarkers were nonnegative, and when a biomarker could
equal zero, we added one to that biomarker so that logs could be taken. Base 2
logs are used, so that a treated-minus-control difference of 1 implies the
treated person's response is twice the matched control's response.
\ Generally, if the difference is $y>0$ then the control's response must be
doubled $y$ times to yield the matched treated subject's response, and if the
the difference is $y<0$ then the treated subject's response must be doubled
$y$ times to yield the matched control's response. \ Figure 2 shows the 234
treated-minus-control matched pair differences, $i=1$, 2, \ldots, $I=234$, in
the logs of the $K=46$ biomarkers, $k=1,\ldots,K=46$. \ A few biomarkers show
large differences.

\begin{figure}[t]
  \centering
  \includegraphics[width = \textwidth]{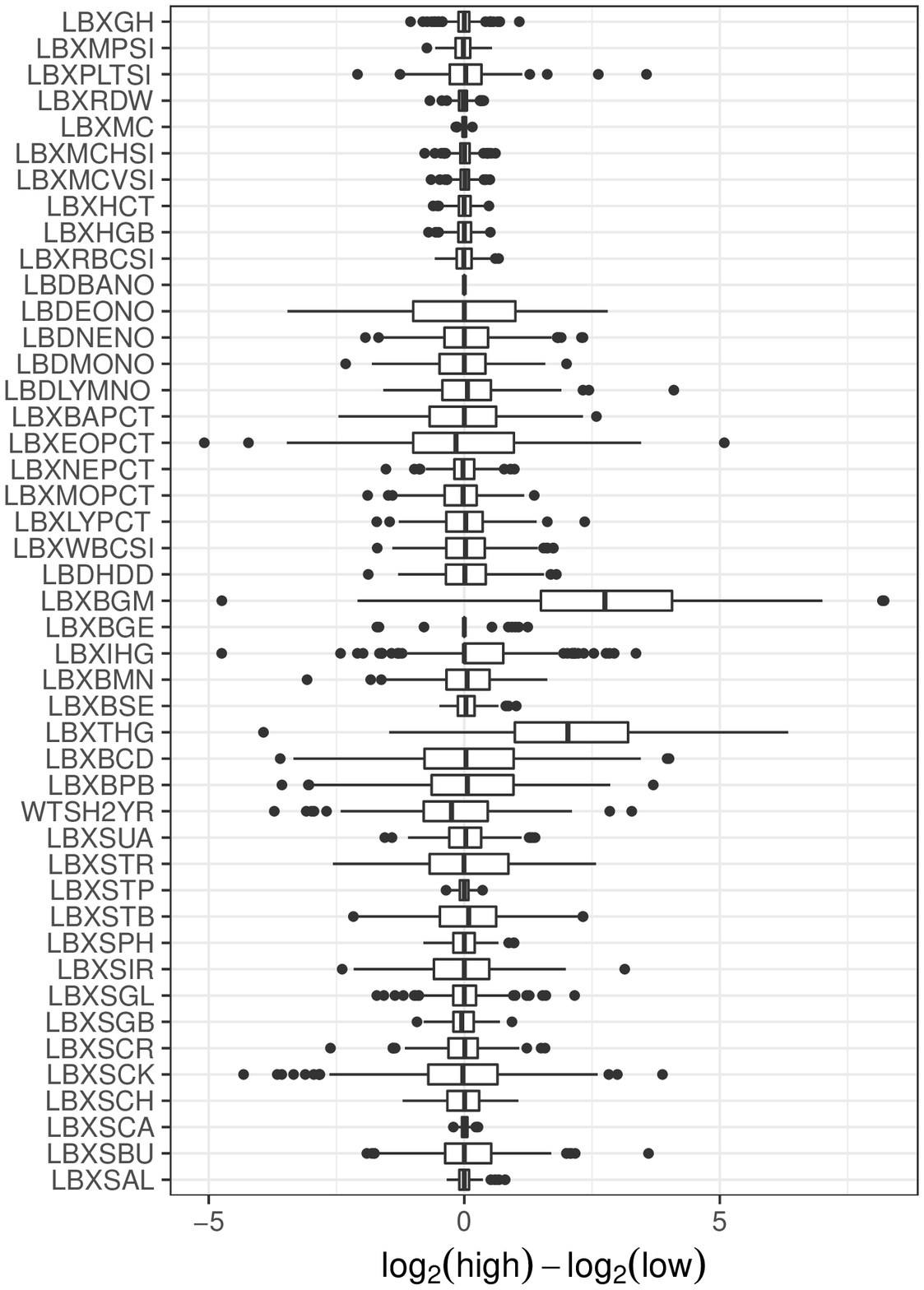}
  \caption{$I = 234$ High-minus-Low matched pair differences,
    $Y_{ik}$, of $\log_{2}(R_{ijk})$ for outcomes $k = 1, \dotsc, 46$. A difference of $1$ on the $\log_2$-scale is $1$ doubling. Note the boxplots for LBXBGM and LBXTHG describing mercury in the blood.}
  \label{fig:2}
\end{figure}

Table
%TCIMACRO{\TeXButton{tabCompare}{\ref{tabCompare}} }%
%BeginExpansion
\ref{tabCompare}
%EndExpansion
compares two multiplicity-corrected sensitivity analyses for the data depicted
in Figure 2. \ Method B applies a Bonferroni correction to Wilcoxon's signed
rank test using the sensitivity analysis that is discussed in Rosenbaum (1987;
2002, \S 4; 2011) and that is reviewed briefly in \S \ref{ssSensitivityReview}%
. \ Specifically, method B performs $2\times46$ one-sided sensitivity
analyses, one in each tail for $K=46$ outcomes, and obtains an upper bound on
each one-sided $P$-value in the presence of a bias of at most $\Gamma$.
\ Here, $\Gamma\geq1$ measures the magnitude of the departure from random
assignment within pairs, such that two matched people might differ in their
odds of treatment by a factor of $\Gamma$ due to some unmatched covariate; see
\S \ref{ssSensitivityReview}. \ The B method takes as the adjusted $P$-value
for a biomarker the smaller of 1 and $2\times K=2\times46$ times the smaller
of the two $P$-value bounds for that biomarker. \ For $\Gamma=1$, this is
simply the Bonferroni correction applied to $K=46$ two-sided Wilcoxon tests,
and for $\left(  \Gamma=1,\,K=1\right)  $ it is simply a two-sided Wilcoxon
test. \ For biomarkers $k=15$, 18, 21, and 23, the corrected $P$-values are
below 0.05 in a randomization test with $\Gamma=1$. \ Biomarker $k=15$ is
sensitive to a nontrivial bias of $\Gamma=1.25$, but biomarkers $k=18$, 21,
and 23 are not. \ Biomarker $k=21$ becomes sensitive at $\Gamma=1.76$ with a
corrected $P$-value bound of 0.054. \ Biomarkers $k=18$ and 23 are sensitive
to a bias of $\Gamma=9$ with $P$-value bounds of 0.095 and 0.075,
respectively, but they are insensitive to a bias of $\Gamma=8$ with $P$-value
bounds of 0.030 and 0.023 (not shown in Table
%TCIMACRO{\TeXButton{tabCompare}{\ref{tabCompare}}}%
%BeginExpansion
\ref{tabCompare}%
%EndExpansion
). \ Perhaps unsurprisingly, the noticeable differences between high and low
consumers of fish are high levels of mercury in the blood of high consumers.

When used in sensitivity analyses, the two-sided Bonferroni method is somewhat
conservative. \ A brief explanation of this follows, but the remainder of this
paragraph is not essential to the paper and may be skipped. \ It is correct
but conservative to say: If the bias in treatment assignment is at most
$\Gamma$, and if $K$ hypotheses are tested, then the probability is at most
$\alpha$ that at least one true null hypothesis yields a one-sided $P$-value
bound less than or equal to $\alpha/\left(  2K\right)  $; see Rosenbaum and
Silber (2009a, \S 4.5). \ The Bonferroni method is conservative for three
reasons. \ First, of course, the Bonferroni inequality is an inequality, not
an equality. \ Second, the one-sided $P$-value bound is obtained by maximizing
the probability in one tail, thereby depleting the opposite tail, so doubling
the one-sided $P$-value bound is conservative for a two-sided test when
$\Gamma>1$. \ Third, when $\Gamma>1$, the worst pattern of biases for one
outcome is often not the worst pattern for other outcomes, and the Bonferroni
method ignores this. \ The Bonferroni-Holm method is least conservative when
different outcomes have effects of very different magnitudes, and it is most
conservative when many weakly correlated outcomes have effects of nearly the
same size. \ Fogarty and Small (2016) propose an optimization technique that
eliminates this third source of conservatism. \ An additional consideration,
not strictly a form of conservatism, is that it may not be best to think about
$K=46$ individual outcomes; rather, there may be greatest insensitivity to
bias for some linear combination of the $K$ outcomes; see Rosenbaum (2016).

The second method in Table
%TCIMACRO{\TeXButton{tabCompare}{\ref{tabCompare}} }%
%BeginExpansion
\ref{tabCompare}
%EndExpansion
is cross-screening (CS), and it is defined in \S \ref{secCrossScreen}.
\ Notably in Table
%TCIMACRO{\TeXButton{tabCompare}{\ref{tabCompare}}}%
%BeginExpansion
\ref{tabCompare}%
%EndExpansion
, cross-screening judges outcomes $k=18$ and 23 to be insensitive to a bias of
$\Gamma=9$ while method B judges them sensitive to $\Gamma=9$; indeed, method
CS judges these two outcomes insensitive to a bias of $\Gamma=11$. \ As
discussed later, we expect cross-screening to lose to method B in a
randomization test, $\Gamma=1$, or when $I/K$ is large. \ We offer reasons to
expect cross-screening to perform well, as it does in Table
%TCIMACRO{\TeXButton{tabCompare}{\ref{tabCompare}}}%
%BeginExpansion
\ref{tabCompare}%
%EndExpansion
, when allowance is made for a nontrivial bias, $\Gamma\geq1.25$, and $I/K$ is
not large.

\section{What is random cross-screening?}

\label{secCrossScreen}

\subsection{Definition and control of the family-wise error rate}

\label{ssCrossScreenDescribe}

Random cross-screening begins by splitting the data set in half at random,
yielding two groups of $I/2$ pairs, or $234/2=117$ pairs in Figure 2. \ The
unusual aspect of cross-screening is that both halves are used to screen and
both halves are used to test. \ Indeed, each half suggests a way to test in
the other half, suggesting a one-tailed test in a particular direction, and
guiding the choice of test statistic. \ The method strongly controls the
family-wise error rate, so that the probability that it falsely rejects at
least one true null hypothesis is at most $\alpha$ providing the bias in
treatment assignment is at most $\Gamma$. \

Cross-screening may provide a still stronger conclusion in which the two
halves replicate each other, in the sense described by Bogomolov and Heller
(2013). \ Indeed, cross-screening is a modification, at a crucial step, of
their replicability method, but refers to random halves of a single study
rather than to two independent studies by different investigators. \ We first
define random cross-screening in \S \ref{ssCrossScreenDescribe}, discuss
aspects of implementation in \S \ref{ssCrossScreenApects}, and illustrate it
in the example in \S \ref{ssCrossScreenExample}, then discuss its connection
to replicability in \S \ref{ssCrossScreenReplicability}.\ \ Although
intuitively pleasing, the technical value of replication with random halves is
not obvious; however, in \S \ref{secNonrandomCS}, we discuss nonrandom
cross-screening, and in this case a single investigator may attach importance
to replication of two parts of one investigation.

There are $K$ null hypotheses, $H_{k}$, $k=1,\ldots,K$. \ In Figure 2, there
are $K=46$ null hypotheses, the $k$th hypothesis asserting no effect of
treatment on the $k$th biomarker. \ Fix $\alpha$ with $0<\alpha<1$;
conventionally, $\alpha=0.05$. \ If $\mathcal{S}$ is a finite set, write
$\left\vert \mathcal{S}\right\vert $ for the number of elements in
$\mathcal{S}$. \

A mild premise of our discussion of sample splitting is that the $I$ units in
the sample are independent of each other, and that the $K$ hypotheses refer to
either a finite or infinite population containing these $I$ units. \ At the
risk of belaboring this premise, it is useful to mention a few particulars to
avoid any possibility of misunderstanding. \ In \S \ref{ssIntroExample}, there
are $I=234$ matched pairs, distinct pairs are assumed independent, and the
hypotheses refer either to the $I=234$ pairs as a finite population or to a
finite or infinite population from which they were drawn. \ If the $I$ units
are clustered, so that units in distinct clusters are independent, then the
clusters, not the $I$ units, should be split at random into two groups of
clusters, keeping individual clusters intact in one group or the other.
\ Cross-screening is not applicable if the $I$ units are from a single time
series of length $I$. \ Also, the hypotheses are about the population, whether
finite or infinite, in the sense that each $H_{k}$ is either true or false as
a description of the population. \ For example, if $H_{3}$ asserted that at
least one person in the finite population of $I=234$ pairs of people is over
100 years old, then $H_{3}$ is true if at least one person in the $I=234$
pairs of people is over 100 years old; otherwise, it is false. \ When we split
the sample in half at random, $H_{3}$ continues to refer to the $I=234$ pairs
of people; that is, it does not become two new hypotheses, $H_{3}^{^{\prime}}$
and $H_{3}^{^{\prime\prime}}$, about two randomly defined groups of
$234/2=117$ pairs of people, where $H_{3}^{^{\prime}}$ might be true and
$H_{3}^{^{\prime\prime}}$ might be false, depending upon how the random split
came out.

Random cross-screening is defined by the following four steps. \ Essentially,
each random half of the data is extensively used to plan the analysis of the
other half, and then the two separate analyses are combined by correcting for
doing two analyses.

\begin{description}
\item[Step \ 1: ] Split the sample in half at random.

\item[Step 2: ] Use the first half in any way at all to select hypotheses
$\mathcal{H}_{2}\subseteq\left\{  1,\ldots,K\right\}  $ to test using the
second half, with $1\leq K_{2}=\left\vert \mathcal{H}_{2}\right\vert \leq K$.
\ At the same time, use the first half in any way at all to select a test
statistic $T_{2k}$ to use in the second half in testing $H_{k}$ for
$k\in\mathcal{H}_{2}$. \ In the second half, use some method for testing the
hypotheses in $\mathcal{H}_{2}$ that would strongly control the family-wise
error rate at $\alpha/2$ if only the $K_{2}$ hypotheses in $\mathcal{H}_{2}$
were tested using $T_{2k}$ for $k\in\mathcal{H}_{2}$. \ Let $\mathcal{R}%
_{2}\subseteq\mathcal{H}_{2}$ be the set of rejected hypotheses.

\item[Step 3: ] Use the second half in any way at all to select hypotheses
$\mathcal{H}_{1}\subseteq\left\{  1,\ldots,K\right\}  $ to test using the
first half, with $1\leq K_{1}=\left\vert \mathcal{H}_{1}\right\vert \leq K$.
At the same time, use the second half in any way at all to select a test
statistic $T_{1k}$ to use in the second half in testing $H_{k}$ for
$k\in\mathcal{H}_{1}$. \ In the first half, use some method for testing the
hypotheses in $\mathcal{H}_{1}$ that would strongly control the family-wise
error rate at $\alpha/2$ if only the $K_{1}$ hypotheses in $\mathcal{H}_{1}$
were tested using $T_{1k}$ for $k\in\mathcal{H}_{1}$. \ Let $\mathcal{R}%
_{1}\subseteq\mathcal{H}_{1}$ be the set of rejected hypotheses.

\item[Step 4: ] Reject $H_{k}$ if $k\in\mathcal{R}=\mathcal{R}_{1}%
\cup\mathcal{R}_{2}$.
\end{description}

It is not novel to split a sample at random into two independent subsamples,
plan the study using the first subsample, carry out the plan with the
independent second subsample,\ viewing decisions from the first subsample as
fixed; see, for instance, Cox (1975) or Heller et al. (2009). \ That is, Step
2 is not novel on its own, nor is Step 3 on its own. \ The novel element in
cross-screening is that this process is done twice, with each subsample
playing both roles, and if either version rejects $H_{k}$ then Step 4 rejects
$H_{k}$. \ In a limited sense, cross-screening uses all of the data to plan
and all of the data to test.

The basic property of cross-screening is that it strongly controls the
family-wise error rate, as discussed in the following proposition. \ Strong
control of the family-wise error rate means that the probability of falsely
rejecting at least one true null hypothesis is at most $\alpha$ no matter
which of the $K$ hypotheses are true. \ Let $\mathcal{T}\subseteq\left\{
1,\ldots,K\right\}  $ be the unknown, possibly empty, set of indices such that
$H_{k}$ is true if and only if $k\in\mathcal{T}$.

\begin{proposition}
\label{PropStrongControl} At least one true null hypothesis, $H_{k}$ with
$k\in\mathcal{T}$, is falsely rejected by cross-screening, with $k\in
\mathcal{R}$, with probability $\Pr\left(  \mathcal{T}\cap\mathcal{R}%
\neq\emptyset\right)  \leq\alpha$.
\end{proposition}

\begin{proof}
Write $\mathcal{A}_{1}$ for all of the observed data in the first half sample,
and $\mathcal{A}_{2}$ for all of the observed data in the second half sample.
\ From Step 2 of cross-screening, $\Pr\left(  \left.  \mathcal{T}%
\cap\mathcal{R}_{2}\neq\emptyset\,\right\vert \,\mathcal{A}_{1}\right)
\leq\alpha/2$, so that $\Pr\left(  \mathcal{T}\cap\mathcal{R}_{2}\neq
\emptyset\right)  =\mathrm{E}\left\{  \Pr\left(  \left.  \mathcal{T}%
\cap\mathcal{R}_{2}\neq\emptyset\,\right\vert \,\mathcal{A}_{1}\right)
\right\}  \leq\alpha/2$. \ In parallel, from Step 3, $\Pr\left(
\mathcal{T}\cap\mathcal{R}_{1}\neq\emptyset\right)  \leq\alpha/2$. \ By the
Bonferroni inequality,
\[
\Pr\left(  \mathcal{T}\cap\mathcal{R}\neq\emptyset\right)  =\Pr\left\{
\left(  \mathcal{T}\cap\mathcal{R}_{1}\neq\emptyset\right)  \;\mathrm{or}%
\;\left(  \mathcal{T}\cap\mathcal{R}_{2}\neq\emptyset\right)  \right\}
\]%
\begin{equation}
\leq\Pr\left(  \mathcal{T}\cap\mathcal{R}_{1}\neq\emptyset\right)  +\Pr\left(
\mathcal{T}\cap\mathcal{R}_{2}\neq\emptyset\right)  \leq\frac{\alpha}{2}%
+\frac{\alpha}{2}=\alpha\text{.} \label{eqPropStrongInequality}%
\end{equation}

\end{proof}

Could use of the Bonferroni inequality in (\ref{eqPropStrongInequality}) be
replaced by an appeal to the independence of the two half-samples, so that
$\Pr\left(  \mathcal{T}\cap\mathcal{R}\neq\emptyset\right)  =1-\left(
1-\alpha\right)  ^{2}$ instead of $\Pr\left(  \mathcal{T}\cap\mathcal{R}%
\neq\emptyset\right)  \leq\alpha$? \ It cannot. \ Although the two half
samples are independent of each other, the event $\left(  \mathcal{T}%
\cap\mathcal{R}_{1}\neq\emptyset\right)  $ depends upon both halves, as does
the event $\left(  \mathcal{T}\cap\mathcal{R}_{2}\neq\emptyset\right)  $, as
seen in Steps 2 and 3 of the description of cross-screening. \ The choices
about which hypotheses to test in the second sample, $\mathcal{H}_{2}$, and
the choice of test statistics, $T_{2k}$, to use when testing hypotheses in
$\mathcal{H}_{2}$, were based on the first sample. \ However, conditionally
given those choices, the chance of at least one false rejection in
$\mathcal{R}_{2}$ is at most $\alpha/2$, because the two halves are
independent. \ The situation is parallel for $\mathcal{R}_{1}$.

\subsection{Aspects of implementing cross-screening}

\label{ssCrossScreenApects}

In Step 2 of cross-screening, we would typically uses the first half sample to
select a small number $K_{2}=\left\vert \mathcal{H}_{2}\right\vert $ of
hypotheses, perhaps even $K_{2}=1$, that appear most likely to be rejected
when tested in the second half sample. \ Also, we would pick test statistics,
$T_{2k}$ for $k\in\mathcal{H}_{2}$, that are mostly likely to reject those
hypotheses. \ A simple way to do this is to test every hypothesis
$k\in\left\{  1,\ldots,K\right\}  $ using several test statistics in the first
half sample, and then pick for use in the second sample the hypotheses and
test statistics with the smallest $P$-values in the first sample. \ Analogous
considerations apply to Step 3 of cross-screening.

In particular, cross-screening permits a few one-sided tests to be performed,
rather than many two-sided tests. \ Shaffer (1974) and Cox (1977, \S 4.2)
discuss reasons for viewing a two-sided test of $H_{k}$ as two one-sided tests
of $H_{k}$ with a Bonferroni correction for testing $H_{k}$ twice, say by
rejecting $H_{k}$ if either $T_{2k}$ or $-T_{2k}$ is large. \ In
cross-screening, we might decide in Step 2 to use either $T_{2k}$ or $-T_{2k}%
$, not both, as our test statistic in the second half based on what we
observed in the first half, thereby omitting the 2-fold Bonferroni correction
needed for a two-sided test. \ In \S \ref{ssIntroExample}, if the first half
sample suggested that eating lots of fish increased a particular biomarker
with a long-tailed distribution, we might select that biomarker for testing in
the second sample, using a one-sided robust test, looking only for an increase
in that biomarker in the second sample.

Suppose that we always picked a single hypothesis to test, so that
$1=K_{1}=\left\vert \mathcal{H}_{1}\right\vert =K_{2}=\left\vert
\mathcal{H}_{2}\right\vert $. \ Then the two halves might select two different
hypotheses, $\mathcal{H}_{1}\neq\mathcal{H}_{2}$, and either or both
hypotheses might be rejected in Step 4 of cross-screening.

A reasonable strategy is to insist that $K_{1}$ and $K_{2}$ each be at least 1
but much less than $K$. \ There is little or no hope of outperforming the
Bonferroni correction for $K$ hypotheses if $K_{1}$ and $K_{2}$ are near $K$.
\ Even if no hypothesis looks especially promising based on the first half
sample, the power of the overall procedure can only be hurt by taking
$0=K_{2}=\left\vert \mathcal{H}_{2}\right\vert $. \ In Table
%TCIMACRO{\TeXButton{tabCompare}{\ref{tabCompare}}}%
%BeginExpansion
\ref{tabCompare}%
%EndExpansion
, we set $K_{1}=K_{2}=2$, selecting the two least sensitive hypotheses in one
half-sample for testing in the other half-sample, using the Bonferroni
inequality in Steps 2 and 3 to strongly control the family-wise error rate
when testing $K_{1}=K_{2}=2$ hypotheses. \ As it turned out, both half samples
selected hypotheses $H_{18}$ and $H_{23}$ for testing in the complementary
half; that is, $\mathcal{H}_{1}=\mathcal{H}_{2}=\left\{  18,\,23\right\}  $.
\ One could instead use Holm's (1979) method for two hypotheses in Steps 2 and
3, thereby gaining in power.

A better approach does not select hypotheses, but rather orders them, and
tests them in sequence by a method that controls the family-wise error rate.
\ We recommend ordering hypotheses so that the least sensitive hypothesis
appears first. \ In the planning half, one performs a sensitivity analysis for
each outcome, determines the $\Gamma$ at which this outcome becomes sensitive
to bias, and orders hypotheses in decreasing order of sensitivity to bias.
\ This approach avoids an arbitrary decision about how many hypotheses to
test, but it emphasizes the most promising hypotheses. \ For methods that
control the family-wise error rate at $\alpha$ when testing hypotheses in a
given order, see, for instance:\ Gansky and Koch (1996), Hsu and Berger
(1999), Wiens (2003), Hommel and Kropf (2005), Rosenbaum (2008) and Burman,
Sonesson and Guilbaud (2009). \ The simplest approach --- so-called fixed
sequence testing --- tests the first hypothesis at level $\alpha$, stops
testing if this hypothesis is accepted, otherwise tests the second hypothesis
at level $\alpha$, continuing until the first acceptance; see, for instance,
Gansky and Koch (1996). \ Wiens (2003) and Hommel and Kropf (2005) test at a
level below $\alpha$, perhaps at $\alpha/2$, so they can continue beyond the
first acceptance; moreover, they transfer unspent $\alpha$ forward to later
hypotheses. \ Wiens calls this a fall-back procedure. \ Burman et al. (2009)
extend that strategy, recycling some unspent $\alpha$ backwards to previously
accepted hypotheses. \ Our simulation in \S \ref{ssSimulation} evaluates these methods.

The hypotheses $\mathcal{H}_{2}$ and test statistics $T_{2k}$ picked in Step 2
must be a function of the first half sample, without input from the second
half sample. \ Otherwise, treating $\mathcal{H}_{2}$ and $T_{2k}$ as fixed
when testing in the second sample would not be the same as conditioning on
their observed values, a key part of Proposition \ref{PropStrongControl}. \ In
parallel, the hypotheses $\mathcal{H}_{1}$ and test statistics $T_{1k}$ picked
in Step 3 must be a function of the second half sample, without input from the
first half sample. \ The most natural, the most convenient, and the most
public way to do this is with an explicit function or algorithm that makes
these choices, and we recommend that approach in practice. \ Notice, however,
that if, without communicating, one coauthor informally used the first half
sample to pick $\mathcal{H}_{2}$ and test statistics $T_{2k}$, while a second
coauthor informally used the second half sample to pick $\mathcal{H}_{1}$ and
test statistics $T_{1k}$, then Proposition \ref{PropStrongControl} would
continue to hold.

\subsection{Cross-screening in the example}

\label{ssCrossScreenExample}

Consider blood mercury (LBXTHG or $k=18$) in Table
%TCIMACRO{\TeXButton{tabCompare}{\ref{tabCompare}}}%
%BeginExpansion
\ref{tabCompare}%
%EndExpansion
. \ In a conventional randomization test, $\Gamma=1$, all of the $P$-values
for $H_{18}$ are extremely small and both B and CS reject the hypothesis of no
effect. \ At $\Gamma=9$, the one-sided Wilcoxon test using all of the data has
$P$-value bound 0.001036 and multiplying this by $2\times46$ yields 0.095, so
method B does not reject in row 18 of in Table
%TCIMACRO{\TeXButton{tabCompare}{\ref{tabCompare}}}%
%BeginExpansion
\ref{tabCompare}%
%EndExpansion
.

The illustration of cross-screening in Table
%TCIMACRO{\TeXButton{tabCompare}{\ref{tabCompare}} }%
%BeginExpansion
\ref{tabCompare}
%EndExpansion
selects the $K_{1}=K_{2}=2$ least sensitive hypotheses in each half-sample for
testing in the complementary half sample. \ As it turns out, both half-samples
selected hypotheses $H_{18}$ and $H_{23}$ for testing, both referring to
mercury levels in the blood. \ Later, in the simulation in
\S \ref{ssSimulation}, we consider methods that do not fix the number of
hypotheses to be tested.

In cross-screening, one-sided tests are performed, with the direction of the
alternative being determined by the complementary half sample. \ Additionally,
two test statistics were tried in one-half sample, and one was selected for
testing in the complementary half. \ One test was Wilcoxon's signed rank test;
the other was a U-statistic called (8,5,8) proposed in Rosenbaum (2011) and
discussed further in \S \ref{secNotationReview}. \ For hypothesis $H_{18}$,
one-sided $P$-value bounds from Wilcoxon's statistic at $\Gamma=9$ are,
respectively, 0.03445 and 0.00647 in the first and second half samples. \ The
one-sided $P$-value bounds from the U-statistic (8,5,8) are, respectively,
0.02132 and 0.00383 in the first and second half samples. \ The quoted
$P$-value bound for $k=18$, $\Gamma=9$ in Table
%TCIMACRO{\TeXButton{tabCompare}{\ref{tabCompare}} }%
%BeginExpansion
\ref{tabCompare}
%EndExpansion
is four times the smallest of these, $0.015=4\times0.00383$. \ This is
because, in Step 2, the first sample suggested testing $H_{18}$ and $H_{23}$
in the second sample. Moreover, in Step 2, the first sample suggested using
the U-statistic in the second sample because $0.03445>0.02132$. \ In the
second sample, these two hypotheses were tested, and Step 2 requires that the
family-wise error in these two tests be controlled at $\alpha/2$, so each test
was performed at level $\alpha/4=\left(  \alpha/2\right)  /2$. \ The smallest
level $\alpha$ that leads to rejection is $0.015=4\times0.00383$, yielding the
quoted $P$-value from $\mathcal{R}_{2}$. \ It is beside the point, but
nonetheless interesting that the second half sample would also have selected
hypotheses $H_{18}$ and $H_{23}$, and would have recommended testing $H_{18}$
with the U-statistic, and it would have yielded a $P$-value bound of
$0.0259=4\times0.00647$ from $\mathcal{R}_{1}$. \ So, at $\Gamma=9$, the two
halves concur in rejecting $H_{18}$ at the 0.05 level, but we only needed
rejection in either half in Step 4. \ In contrast, at $\Gamma=11$, both half
samples recommend testing $H_{18}$ and $H_{23}$ using the U-statistic; however
the $P$-value bounds for $H_{18}$ are 0.04589 and 0.00865 in the two halves,
so that only the second half sample, $\mathcal{R}_{2}$, leads to rejection
with $0.05>0.035=4\times0.00865$. \

This $P$-value bound of 0.035 from CS should be compared with the $P$-value
bound of 0.505 from the Bonferroni method at $\Gamma=11$. \ The Bonferroni
method corrects for $K=46$ two-sided tests, while cross-screening corrects for
4 one-sided tests, two in each of two half samples. \ Additionally,
cross-screening adaptively selected for use in one half sample the test
statistic that performed best in the other half sample.

\subsection{Cross-screening and replicability}

\label{ssCrossScreenReplicability}

With a different goal in view, Bogomolov and Heller (2013) used a procedure
similar in form to Steps 1 to 3 of cross-screening, but with a different Step
4, and with some other differences. \ They asked whether two independent
studies by different investigators testing the same $K$ hypotheses have
replicated each other. \ Their method says that rejection of $H_{k}$ has
replicated if $k\in\mathcal{R}_{1}\cap\mathcal{R}_{2}$, whereas Step 4 rejects
$H_{k}$ in one study if $k\in\mathcal{R}_{1}\cup\mathcal{R}_{2}$; that is,
both studies must reject $H_{k}$ to replicate, rather than either half
rejecting $H_{k}$ to control the family-wise error rate in a single study.
\ Obviously, $\left(  \mathcal{R}_{1}\cap\mathcal{R}_{2}\right)
\subseteq\left(  \mathcal{R}_{1}\cup\mathcal{R}_{2}\right)  $ so the
probability of replication is lower than the probability of rejection in
cross-screening. \ In the numerical example in \S \ref{ssCrossScreenExample},
hypothesis $k=18\in$ $\mathcal{R}_{1}\cap\mathcal{R}_{2}$ for $\Gamma=9$,
thereby meeting Bogomolov and Heller (2013)'s standard for replication; at
$\Gamma=11$, however, $18\notin$ $\mathcal{R}_{1}\cap\mathcal{R}_{2}$ but
$18\in\mathcal{R}_{1}\cup\mathcal{R}_{2}$, so $H_{18}$ is rejected by
cross-screening without meeting the standard for replication.

As Bogomolov and Heller (2013) emphasize, when asking about replicability of
two independent studies, it is possible that $H_{k}$ is genuinely true in one
study and genuinely false in the other: there might be a treatment effect for
outcome $k$ in Cleveland and not in Kyoto, or a specific gene variant might
have an effect in a study done on one ethnic group, but not in a study done on
a different ethnic group. \ In contrast, with a random split of a random
sample from a single population, $H_{k}$ is either true in the population or
false in the population. \ Additionally, Steps 2 and 3 of cross-screening
treat the two half samples symmetrically, each informing the analysis of the
other, for instance the choice of test statistic, or the choice of tail for a
two-sided test. \ With two independent studies by different investigators, it
is likely that one study came first, so symmetry is unlikely. \ Perhaps the
earlier study influenced the design and analysis of the later study, but the
symmetry of Steps 2 and 3 is not expected when one study follows another.

Section \ref{secNonrandomCS} discusses an intermediate case involving
nonrandom cross-screening in a single study. \ With nonrandom cross-screening
in \S \ref{secNonrandomCS}, it is interesting to know both whether
$k\in\mathcal{R}_{1}\cap\mathcal{R}_{2}\ $\ and also whether $k\in
\mathcal{R}_{1}\cup\mathcal{R}_{2}$.

\section{Notation for paired randomized experiments and observational studies}

\label{secNotationReview}

\subsection{Causal inference in paired randomized experiments}

\label{ssRandomizationReview}

Cross-screening, as described in \S \ref{secCrossScreen}, is quite general.
\ It could be used with matched pairs, matched sets with multiple controls,
full matching, unmatched comparisons, with cohort or case-control studies.
\ The simplest case, however, concerns treatment-control matched pairs, as in
\S \ref{ssIntroExample}, and the evaluation in \S \ref{secEvaluation} will be
restricted to this simplest case.

There are $I$ matched pairs, $i=1,\ldots,I$, of two individuals, $j=1,2$, one
treated with $Z_{ij}=1$, the other an untreated control with $Z_{ij}=0$, so
$1=Z_{i1}+Z_{i2}$ for each $i$. \ Individuals were matched for an observed
covariate $x_{ij}$, so that $x_{i1}=x_{i2}$ for each $i$, but there is concern
about an unmeasured covariate $u_{ij}$ for which it is possible that
$u_{i1}\neq u_{i2}$ for many or all pairs $i$. \ Write $\mathbf{Z}=\left(
Z_{11},\ldots,Z_{I2}\right)  ^{T}$ for the $2I$ dimensional vector of
treatment assignments, and let $\mathcal{Z}$ be the set containing the $2^{I}$
possible values $\mathbf{z}$ of $\mathbf{Z}$, so $\mathbf{z}\in\mathcal{Z}$ if
$\mathbf{z}=\left(  z_{11},\ldots,z_{I2}\right)  ^{T}$ with $z_{ij}=1$ or $0$
and $z_{i1}+z_{i2}=1$ for all $i$, $j$. \ Conditioning on the event
$\mathbf{Z}\in\mathcal{Z}$ is abbreviated as conditioning on $\mathcal{Z}$.

There are $K$ responses, $k=1,\ldots,K$. \ The $j$th individual in pair $i$
has two potential outcomes for the $k$th response, namely $r_{Tijk}$ if this
individual is treated with $Z_{ij}=1$ or $r_{Cijk}$ if this individual is
given control with $Z_{ij}=0$, so the $k$th response observed from this
individual is $R_{ijk}=Z_{ij}\,r_{Tijk}+\left(  1-Z_{ij}\right)  \,r_{Cijk}$
and the effect caused by the treatment, namely $\delta_{ijk}=r_{Tijk}%
-r_{Cijk}$, is not observed for any individual; see Neyman (1923) and Rubin
(1974). \ Fisher's (1935) sharp null hypothesis of no treatment effect for
response $k$ asserts $H_{k}:\delta_{ijk}=0$ for all $i$ and $j$. \ Write
$\mathbf{R}_{k}=\left(  R_{11k},\ldots,R_{I2k}\right)  ^{T}$ and
$\mathbf{r}_{Ck}=\left(  r_{C11k},\ldots,r_{CI2k}\right)  ^{T}$, so that
$\mathbf{R}_{k}=\mathbf{r}_{Ck}$ when $H_{k}$ is true. \ For the potential
outcomes and covariates for the $2I$ individuals, write $\mathcal{F}=\left[
\left\{  \left(  r_{Tijk},\,r_{Cijk},\,k=1,\ldots,K\right)  ,\,x_{ij}%
,\,u_{ij}\right\}  ,\,i=1,\ldots,I,\,j=1,2\right]  $.

In a paired, randomized experiment, a fair coin is flipped independently $I$
times to assign treatments within the $I$ pairs, so that $\Pr\left(  \left.
\mathbf{Z}=\mathbf{z\,}\right\vert \,\mathcal{F},\,\mathcal{Z}\right)
=2^{-I}=\left\vert \mathcal{Z}\right\vert ^{-1}$ for each $\mathbf{z}%
\in\mathcal{Z}$. \ A randomization test of $H_{k}$ compares the distribution
of a test statistic, $T_{k}=t_{k}\left(  \mathbf{Z},\mathbf{R}_{k}\right)  $,
to its randomization distribution, $\Pr\left(  \left.  T_{k}\geq
t\mathbf{\,}\right\vert \,\mathcal{F},\,\mathcal{Z}\right)  $, when $H_{k}$ is
true, namely%
\begin{equation}
\Pr\left\{  \left.  t_{k}\left(  \mathbf{Z},\mathbf{R}_{k}\right)  \geq
t\mathbf{\,}\right\vert \,\mathcal{F},\,\mathcal{Z}\right\}  =\Pr\left\{
\left.  t_{k}\left(  \mathbf{Z},\mathbf{r}_{Ck}\right)  \geq t\mathbf{\,}%
\right\vert \,\mathcal{F},\,\mathcal{Z}\right\}  =\frac{\left\vert \left\{
\mathbf{z}\in\mathcal{Z}:t_{k}\left(  \mathbf{z},\mathbf{r}_{Ck}\right)  \geq
t\right\}  \right\vert }{\left\vert \mathcal{Z}\right\vert }\text{,}
\label{eqNullRand}%
\end{equation}
because $\mathbf{R}_{k}=\mathbf{r}_{Ck}$ when $H_{k}$ is true, $\mathbf{r}%
_{Ck}$ is fixed by conditioning on $\mathcal{F}$, and $\Pr\left(  \left.
\mathbf{Z}=\mathbf{z\,}\right\vert \,\mathcal{F},\,\mathcal{Z}\right)
=\left\vert \mathcal{Z}\right\vert ^{-1}$ in a randomized experiment.

A common form of test statistic $T_{k}=t_{k}\left(  \mathbf{Z},\mathbf{R}%
_{k}\right)  $ uses the treated-minus-control matched-pair difference,
$Y_{ik}=\left(  Z_{i1}-Z_{i2}\right)  \left(  R_{i1k}-R_{i2k}\right)  $, which
equals $\pm\left(  r_{Ci1k}-r_{Ci2k}\right)  $ when $H_{k}$ is true. \ The
absolute pair differences, $\left\vert Y_{ik}\right\vert $, are assigned
nonnegative scores $q_{ik}\geq0$ with $q_{ik}=0$ when $\left\vert
Y_{ik}\right\vert =0$. \ Write $\mathrm{sgn}\left(  y\right)  =1$ if $y>0$ and
$\mathrm{sgn}\left(  y\right)  =0$ otherwise. \ The randomization distribution
of statistics of the form $T_{k}=\sum_{i=1}^{I}\,\mathrm{sgn}\left(
Y_{ik}\right)  \;q_{ik}$ yield many familiar randomization tests, including:
(i) Wilcoxon's signed rank test with $q_{ik}$ to be the rank of $\left\vert
Y_{ik}\right\vert $, (ii) the permutational $t$-test with $q_{ik}=\left\vert
Y_{ik}\right\vert /I$, (iii) Maritz (1979)'s randomization distribution for
Huber's $M$-statistic with $q_{ik}=\psi\left(  \left\vert Y_{ik}\right\vert
/\varkappa\right)  $\ where $\varkappa$ is the median $\left\vert
Y_{ik}\right\vert $ and $\psi\left(  \cdot\right)  $ is a monotone increasing
odd function, $\psi\left(  y\right)  =-\psi\left(  -y\right)  $. \ Under
$H_{k}$, the difference $R_{i1k}-R_{i2k}=r_{Ci1k}-r_{Ci2k}$ is fixed by
conditioning on $\mathcal{F}$, so the null distribution (\ref{eqNullRand}) of
$T_{k}$ is the distribution of the sum of $I$ independent random variables
taking the values $q_{ik}$ or 0 with equal probabilities $1/2$ if $\left\vert
Y_{ik}\right\vert >0$ or the value 0 with probability 1 if $\left\vert
Y_{ik}\right\vert =0$.

The U-statistic $\left(  m,\underline{m},\overline{m}\right)  =\left(
8,5,8\right)  $ was used in cross-screening in Table
%TCIMACRO{\TeXButton{tabCompare}{\ref{tabCompare}}}%
%BeginExpansion
\ref{tabCompare}%
%EndExpansion
. \ Let $a_{ik}$ be the rank of $\left\vert Y_{ik}\right\vert $. \ In general,
the U-statistic $\left(  m,\underline{m},\overline{m}\right)  $ has
$q_{ik}=\binom{I}{m}^{-1}\sum_{\ell=\underline{m}}^{\overline{m}}\binom
{a_{ik}-1}{\ell-1}\binom{I-a_{ik}}{m-\ell}$; see Rosenbaum (2011). \ Here,
$\left(  m,\underline{m},\overline{m}\right)  =\left(  1,1,1\right)  $ yields
the sign test, $\left(  2,2,2\right)  $ is nearly the same as Wilcoxon's test,
$\left(  m,\underline{m},\overline{m}\right)  =\left(  m,m,m\right)  $ for
integer $m\geq3$ yields Stephenson's (1981) test, and $\left(  8,5,8\right)  $
is an $S$-shaped transformation of the ranks that diminishes the influence of
$Y_{ik}$ that are close to zero.

\subsection{Sensitivity analysis in observational studies}

\label{ssSensitivityReview}

A simple model for biased treatment assignment in observational studies
asserts that in the population, prior to matching, treatment assignments are
independent with probabilities $\pi_{ij}=\Pr\left(  \left.  Z_{ij}%
=1\mathbf{\,}\right\vert \,\mathcal{F}\right)  $, and two subjects, say $ij$
and $i^{\prime}j^{\prime}$, with the same value of $x$, $x_{ij}=x_{i^{\prime
}j^{\prime}}$, may differ in their odds of treatment by at most a factor of
$\Gamma\geq1$, that is $\Gamma^{-1}\leq\pi_{ij}\left(  1-\pi_{i^{\prime
}j^{\prime}}\right)  /\left\{  \pi_{i^{\prime}j^{\prime}}\left(  1-\pi
_{ij}\right)  \right\}  \leq\Gamma$, and then conditions on $\mathbf{Z}%
\in\mathcal{Z}$. \ Setting $\gamma=\log\left(  \Gamma\right)  \geq0$, this is
equivalent to introducing an unobserved covariate $\mathbf{u}=\left(
u_{11},\ldots,u_{I2}\right)  ^{T}$ with $0\leq u_{ij}\leq1$ such that
$\Pr\left(  \left.  \mathbf{Z}=\mathbf{z\,}\right\vert \,\mathcal{F}%
,\,\mathcal{Z}\right)  =\exp\left(  \gamma\mathbf{u}^{T}\mathbf{z}\right)
/\sum_{\mathbf{b}\in\mathcal{Z}}\exp\left(  \gamma\mathbf{u}^{T}%
\mathbf{b}\right)  $ for $\mathbf{z}\in\mathcal{Z}$, or equivalently
\begin{equation}
\Pr\left(  \left.  \mathbf{Z}=\mathbf{z\,}\right\vert \,\mathcal{F}%
,\,\mathcal{Z}\right)  =\prod\nolimits_{i=1}^{I}\theta_{i}^{z_{i1}}\,\left(
1-\theta_{i}\right)  ^{z_{i2}}\text{ with }\frac{1}{1+\Gamma}\leq\theta
_{i}\leq\frac{\Gamma}{1+\Gamma}\text{ for each }i\text{;} \label{eqSenMod}%
\end{equation}
see Rosenbaum (2002, \S 4) where equivalences are demonstrated by deriving
$u_{ij}$ from $\pi_{ij}$ and conversely. \ Define\ $\overline{\overline{T}%
}_{\Gamma k}$ to be the sum of $I$ independent random variables taking the
value $q_{ik}\geq0$ with probability $\kappa=\Gamma/\left(  1+\Gamma\right)  $
or the value $0$ with probability $1/\left(  1+\Gamma\right)  $, and define
$\overline{T}_{\Gamma k}$ analogously but with the probabilities reversed.
\ Then for $T_{k}=\sum_{i=1}^{I}\,\mathrm{sgn}\left(  Y_{ik}\right)  \,q_{ik}%
$, it is not difficult to show under $H_{k}$ and (\ref{eqSenMod}) that the
null distribution of $T_{k}$ is bounded by two known distributions,
\begin{equation}
\Pr\left(  \left.  \overline{T}_{\Gamma k}\geq t\mathbf{\,}\right\vert
\,\mathcal{F},\,\mathcal{Z}\right)  \leq\Pr\left(  \left.  T_{k}\geq
t\mathbf{\,}\right\vert \,\mathcal{F},\,\mathcal{Z}\right)  \leq\Pr\left(
\left.  \overline{\overline{T}}_{\Gamma k}\geq t\mathbf{\,}\right\vert
\,\mathcal{F},\,\mathcal{Z}\right)  \text{,} \label{eqSenMod2}%
\end{equation}
from which bounds on $P$-values, point estimates and confidence intervals are
obtained. \ In particular, the upper bound on the one-sided $P$-value is
obtained by evaluating $\Pr\left(  \left.  \overline{\overline{T}}_{\Gamma
k}\geq t\mathbf{\,}\right\vert \,\mathcal{F},\,\mathcal{Z}\right)  $ in
(\ref{eqSenMod2}) with $t$ set to the observed value of the test statistic,
$T_{k}$. \ Under mild conditions on the scores, $q_{ik}$, as $I\rightarrow
\infty$, the upper bound in (\ref{eqSenMod2}) can be approximated by%
\begin{equation}
\Pr\left(  \left.  \overline{\overline{T}}_{\Gamma k}\geq t\mathbf{\,}%
\right\vert \,\mathcal{F},\,\mathcal{Z}\right)  \doteq1-\Phi\left(
\frac{t-\kappa\sum_{i=1}^{I}q_{ik}}{\sqrt{\kappa\left(  1-\kappa\right)
\sum_{i=1}^{I}q_{ik}^{2}}}\right)  \label{eqNormalApproxSen}%
\end{equation}
where $\Phi\left(  \cdot\right)  $ is the standard Normal cumulative distribution.

\section{Evaluation of random cross-screening}

\label{secEvaluation}

\subsection{Random cross-screening performs poorly for $\Gamma=1$ and $K=1$}

\label{ssBadGamma1K1}

Cross-screening can perform poorly. \ Consider the simplest case, namely one
outcome, $K=1$, no bias from unmeasured covariates, $\Gamma=1$, where $Y_{i1}%
$, $i=1,\ldots,I$, are independent observations from a Normal distribution
with expectation $\tau$ and variance 1, $Y_{i1}\sim_{\mathrm{iid}}N\left(
\tau,1\right)  $, testing $H_{1}:\tau=0$ against the alternative that $\tau
>0$. \ In this case, the uniformly most powerful test is based on the mean,
$I^{-1}\sum_{i=1}^{I}Y_{i1}$, the Bonferroni method uses this best test with
no correction because $K=1$, and cross-screening must be inferior. \ It is a
mistake to use cross-screening in a randomized experiment with a small number
of null hypotheses.

More generally, we expect random cross-screening to perform poorly in
comparison with the Bonferroni-Holm method when there are few hypotheses,
$K\leq10$, and the randomization test, $\Gamma=1$, or tiny biases,
$\Gamma=1.1$, are of primary interest. \

The situation can be very different when $K$ is large, perhaps $K\geq100$, and
interest is confined to findings that resist a nontrivial bias, say
$\Gamma\geq1.25$. \ A bias of $\Gamma=1.25$ corresponds with an unobserved
covariate $u$ that doubles the odds of treatment, $Z_{i1}-Z_{i2}=1$, and
doubles the odds of a positive difference in outcomes, $Y_{ik}>0$; see
Rosenbaum and Silber (2009b). \ Reducing $K$ is of substantial value only if
$K$ is fairly large.

\subsection{Random cross-screening and design sensitivity}

\label{ssDesignSen}

If there were a treatment effect and no bias from unmeasured covariates in an
observational study, then the investigator would not be able to recognize this
from the observed data, and the best she could hope to say in this
\textit{favorable situation} is that the conclusions are insensitive to
moderate biases as measured by $\Gamma$. \ In typical situations, as the
number of pairs increases, $I\rightarrow\infty$, in this favorable situation,
the degree of sensitivity to bias tends to a limit, called the design
sensitivity, $\widetilde{\Gamma}$, such that the conclusions are eventually
insensitive to all biases with $\Gamma<\widetilde{\Gamma}$ and sensitive to
some biases with $\Gamma>\widetilde{\Gamma}$; see Rosenbaum (2004; 2010, Part
III). \ More precisely, the upper bound on the $P$-value tends to $1$ as
$I\rightarrow\infty$ for $\Gamma>\widetilde{\Gamma}$, and it tends to $0$ for
$\Gamma<\widetilde{\Gamma}$. \ In that sense, the design sensitivity
$\widetilde{\Gamma}$ is the limiting sensitivity to bias as $I\rightarrow
\infty$ in a favorable situation. \ Moreover, the rate at which the $P$-value
tends to 0 as $I\rightarrow\infty$ for $\Gamma<\widetilde{\Gamma}$ is the
Bahadur efficiency of the sensitivity analysis, and it is useful in
characterizing the performance of alternative methods when $\Gamma
<\widetilde{\Gamma}$; see Rosenbaum (2015). \ The Bahadur efficiency drops to
zero as $\Gamma$ increases to $\widetilde{\Gamma}$.

Generally, the design sensitivity $\widetilde{\Gamma}$ and Bahadur efficiency
depend upon the nature of the favorable situation and on the chosen methods of
analysis. \ A poor choice of test statistic often means an exaggerated report
of sensitivity to unmeasured bias, even in large samples. \ Suppose, for
instance, that $Y_{i1}\sim_{\mathrm{iid}}N\left(  \tau,1\right)  $, and
compare Wilcoxon's signed rank statistic and the U-statistic $\left(
8,5,8\right)  $ that was used in cross-screening in Table
%TCIMACRO{\TeXButton{tabCompare}{\ref{tabCompare}}}%
%BeginExpansion
\ref{tabCompare}%
%EndExpansion
. \ If $\tau=1/2$, then $\widetilde{\Gamma}=3.2$ for Wilcoxon's statistic, but
$\widetilde{\Gamma}=4.2$ for the U-statistic $\left(  8,5,8\right)  $, whereas
for $\tau=1$ the design sensitivity is $\widetilde{\Gamma}=11.7$ for
Wilcoxon's statistic, but $\widetilde{\Gamma}=26.3$ for the U-statistic; see
Rosenbaum (2011, Table 3). \ At $\tau=1/2$, the upper bound on the $P$-value
is tending to 1 as $I\rightarrow\infty$ for $\Gamma=3.5>3.2=\widetilde{\Gamma
}$ if Wilcoxon's statistic is used, but it is tending to zero if $\left(
8,5,8\right)  $ is used because $\Gamma=3.5<4.2=\widetilde{\Gamma}$.
\ Moreover, with $\tau=1/2$ and $\Gamma=2$, the $P$-value bound is tending to
zero at a faster rate for $\left(  8,5,8\right)  $ than for Wilcoxon's
statistic, with Bahadur relative efficiency $>1.25$; see Rosenbaum (2015,
Table 2). \ A similar pattern is found for a shift of $\tau$ with logistic or
$t$-distributed errors. \ In a randomization test with Normal or logistic
errors, the Pitman efficiency of $\left(  8,5,8\right)  $ relative to
Wilcoxon's test is 0.97, so a small loss of efficiency in a randomization
test, $\Gamma=1$, translates into substantial gains in a sensitivity analysis,
$\Gamma>1$. \ Although $\left(  8,5,8\right)  $ is consistently slightly
better than Wilcoxon's statistic for $\Gamma\geq2$, other statistics beat
$\left(  8,5,8\right)  $ for particular error distributions with no uniform
winner; see Rosenbaum (2011, 2015). \ Adaptive choice of a test statistic can
improve design sensitivity and Bahadur efficiency; see Berk and Jones (1978)
and Rosenbaum\ (2012, 2015).

The considerations in the previous paragraph point to one advantage of
cross-screening. \ In Table
%TCIMACRO{\TeXButton{tabCompare}{\ref{tabCompare}}}%
%BeginExpansion
\ref{tabCompare}%
%EndExpansion
, cross-screening adaptively picked either Wilcoxon's statistic or $\left(
8,5,8\right)  $ based on their performance in the complementary half sample.
\ As $I\rightarrow\infty$ in a favorable situation, if two statistics have
different design sensitivities, cross-screening will eventually pick the
statistic with the larger design sensitivity and report greater insensitivity
to unmeasured bias than would have resulted with a fixed but mistaken choice
of statistic. \

The Bonferroni procedure can also be used to adaptively select one of two test
statistics, thereby also attaining the better of two design sensitivities.
\ However, the Bonferroni procedure would have to pay for adaptation by
testing at level $\alpha/\left(  4K\right)  $ rather than $\alpha/\left(
2K\right)  $. \ There are less costly ways to obtain adaptive inferences; see
Rosenbaum (2012).

\subsection{Expected $P$-values and test size in sensitivity analyses with
$\Gamma$ slightly too large}

\label{ssExpectedPvalue}

The Bonferroni method in Table
%TCIMACRO{\TeXButton{tabCompare}{\ref{tabCompare}} }%
%BeginExpansion
\ref{tabCompare}
%EndExpansion
made allowance for a 0.05 chance of a false rejection of a true hypothesis in
$2K=96$ one-sided tests. \ Cross-screening eliminated most hypotheses,
focusing on just two of them. \ Was it wise to do this?

The defining feature of the sensitivity analysis (\ref{eqSenMod2}) testing
$H_{k}$ is that if: (i) $H_{k}$ is true and (ii) the bias in treatment
assignment is at most $\Gamma$, then the chance that the $P$-value bound is
less than or equal to $\alpha$ is at most $\alpha$ for each $0\leq\alpha\leq
1$; that is, the $P$-value bound is stochastically larger than the uniform
distribution. \ The Bonferroni method and most if not all methods that use
$P$-values in controlling the family-wise error rate are designed so that they
give the correct result for $P$-values that are uniformly distributed.
\ However, with $\Gamma>1$, it is common to see many $P$-value bounds that
look strictly larger than the uniform distribution. \

For example, in the calculations used to carry out the cross-screening in
Table
%TCIMACRO{\TeXButton{tabCompare}{\ref{tabCompare}}}%
%BeginExpansion
\ref{tabCompare}%
%EndExpansion
, a total of $2\times2\times2\times46=368$ one-sided $P$-values were examined,
looking at 2 tails, with 2 test statistics, in 2 halves of the data, for 46
outcomes. \ At $\Gamma=1$, $40/368=11\%$ of these $P$-values were $\leq0.05$.
\ However, at $\Gamma=1.25$, only 12 $P$-value bounds (\ref{eqSenMod2}) were
$\leq0.05$, and these occurred for 3 hypotheses $H_{18}$, $H_{21}$, and
$H_{23}$, in both half samples, as judged by both statistics, $12=2\times
2\times3$. \ Even at $\Gamma=1.25$, small $P$-value bounds were rare, and they
may all reflect a genuine effect of eating fish on the level of mercury in the blood.

What happened in Table
%TCIMACRO{\TeXButton{tabCompare}{\ref{tabCompare}} }%
%BeginExpansion
\ref{tabCompare}
%EndExpansion
is not unexpected. \ Suppose $H_{k}$ is true and the bias in treatment
assignment is at most $\Gamma^{\prime}$, but the sensitivity analysis is
performed at $\Gamma>\Gamma^{\prime}$. \ In this case, it is quite unlikely
that a true null hypothesis will yield a small $P$-value bound. \ To
demonstrate this, notice that (\ref{eqSenMod2}) implies the statistic $T_{k}$
is stochastically smaller than $\overline{\overline{T}}_{\Gamma^{\prime}k}$
which is stochastically smaller than $\overline{\overline{T}}_{\Gamma k}$, so
$\Pr\left(  \left.  T_{k}\geq t\mathbf{\,}\right\vert \,\mathcal{F}%
,\,\mathcal{Z}\right)  \leq\Pr\left(  \left.  \overline{\overline{T}}%
_{\Gamma^{\prime}k}\geq t\mathbf{\,}\right\vert \,\mathcal{F},\,\mathcal{Z}%
\right)  \leq\Pr\left(  \left.  \overline{\overline{T}}_{\Gamma k}\geq
t\mathbf{\,}\right\vert \,\mathcal{F},\,\mathcal{Z}\right)  $. \ Using
(\ref{eqNormalApproxSen}) twice at $\Gamma$ and $\Gamma^{\prime}$, as
$I\rightarrow\infty$ we approximate the probability that $\overline
{\overline{T}}_{\Gamma^{\prime}k}$ exceeds the upper $\alpha$ critical value
for $\overline{\overline{T}}_{\Gamma k}$ as%

\begin{equation}
1-\Phi\left\{  \frac{\left(  \kappa-\kappa^{\prime}\right)  \sum_{i=1}%
^{I}q_{ik}+\Phi^{-1}\left(  1-\alpha\right)  \sqrt{\kappa\left(
1-\kappa\right)  \sum_{i=1}^{I}q_{ik}^{2}}}{\sqrt{\kappa^{\prime}\left(
1-\kappa^{\prime}\right)  \sum_{i=1}^{I}q_{ik}^{2}}}\right\}  \text{,}
\label{eqSizeBound}%
\end{equation}
where $\kappa^{\prime}=\Gamma^{\prime}/\left(  1+\Gamma^{\prime}\right)  $.
\ In effect, (\ref{eqSizeBound}) is an upper bound on the size of an $\alpha
$-level sensitivity analysis conducted with the sensitivity parameter $\Gamma$
set above the true bias $\Gamma^{\prime}$. \ Alternatively, when $H_{k}$ is
true, we may approximate the expected $P$-value bound, that is, the
expectation of $\Pr\left(  \left.  \overline{\overline{T}}_{\Gamma k}\geq
t\mathbf{\,}\right\vert \,\mathcal{F},\,\mathcal{Z}\right)  $ when $t$ has the
distribution of $\overline{\overline{T}}_{\Gamma^{\prime}k}$, as expression
(3) in Sackrowitz and Samuel-Cahn (1999),%

\begin{equation}
\mathrm{EPV}=\Phi\left[  \frac{\left(  \kappa-\kappa^{\prime}\right)
\sum_{i=1}^{I}q_{ik}}{\sqrt{\left\{  \kappa\left(  1-\kappa\right)
+\kappa^{\prime}\left(  1-\kappa^{\prime}\right)  \right\}  \sum_{i=1}%
^{I}q_{ik}^{2}}}\right]  \text{;} \label{eqEPV}%
\end{equation}
see also Dempster and Schatzoff (1965). \ For Wilcoxon's signed rank statistic
without ties, $\sum_{i=1}^{I}q_{ik}=I\left(  I+1\right)  /2$ and $\sum
_{i=1}^{I}q_{ik}^{2}=I\left(  I+1\right)  \left(  2I+1\right)  /6$. \ Both
(\ref{eqSizeBound}) and (\ref{eqEPV}) use the large sample approximation to
the distributions of $\overline{\overline{T}}_{\Gamma^{\prime}k}$ and
$\overline{\overline{T}}_{\Gamma k}$, so they make no allowance for the
discreteness of permutation distributions.

Table
%TCIMACRO{\TeXButton{reftabEPV}{\ref{tabEPV}} }%
%BeginExpansion
\ref{tabEPV}
%EndExpansion
evaluates (\ref{eqSizeBound}) and (\ref{eqEPV}) for a one-sided 0.05-level
Wilcoxon test for $I=100$, 250 and 500 pairs, for several values of
$\Gamma^{\prime}\leq\Gamma$. \ Of course, when $\Gamma^{\prime}=\Gamma$, the
size-bound equals the level-bound, 0.05, and the expected $P$-value is 0.5.
\ However, when $\Gamma^{\prime}<\Gamma$, the size is well below 0.05 and the
expected $P$-value is well above 0.5, a pattern that becomes more noticeable
as $I$ increases. \ If $\Gamma$ is a little too large, small $P$-values for
true hypotheses are improbable.

Table
%TCIMACRO{\TeXButton{reftabEPV}{\ref{tabEPV}} }%
%BeginExpansion
\ref{tabEPV}
%EndExpansion
is relevant because it is uncommon to see many outcomes that are equally
sensitive to unmeasured biases; see, for instance, Table
%TCIMACRO{\TeXButton{tabCompare}{\ref{tabCompare}}}%
%BeginExpansion
\ref{tabCompare}%
%EndExpansion
. \ Our tests must allow for this uncommon situation as a logical possibility,
but the Bonferroni method and cross-screening do this in two very different
ways, though both methods strongly control the family-wise error rate. \ The
Bonferroni method pays a price for every $P$-value bound that is computed.
\ Cross-screening with $\Gamma\geq1.25$ does not, in typical situations,
pursue or spend resources on $P$-value bounds that are very large.

\subsection{Trade-off of sample size $I$ and number of hypotheses $K$ when
$\Gamma=1$}

\label{ssTradeoff}

Is it ever worth half the sample to test one hypothesis rather than $K$
hypotheses? \ Table
%TCIMACRO{\TeXButton{reftabttest}{\ref{tabttest}} }%
%BeginExpansion
\ref{tabttest}
%EndExpansion
is several steps removed from cross-screening, but it provides one simple view
of the trade-off of sample size $I$ and the number of hypotheses under test
$K$ when $\Gamma=1$. \ Consider again the case in which $Y_{i1}\sim
_{\mathrm{iid}}N\left(  \tau,1\right)  $ and $Y_{ik}\sim_{\mathrm{iid}%
}N\left(  0,1\right)  $ for $k=2,\ldots,K$, so $H_{1}$ is false and the other
$H_{k}$ are true. \ Table
%TCIMACRO{\TeXButton{reftabttest}{\ref{tabttest}} }%
%BeginExpansion
\ref{tabttest}
%EndExpansion
is simply a sample size calculation for the t-test for hypothesis $H_{1}$ in
the presence of $K-1$ other true null hypotheses.

Table
%TCIMACRO{\TeXButton{reftabttest}{\ref{tabttest}} }%
%BeginExpansion
\ref{tabttest}
%EndExpansion
shows the number of pairs, $I$, required for 80\% power in a single two-sided
0.05-level t-test for $I$ pair differences $Y_{i1}$ that are $Y_{i1}%
\sim_{\mathrm{iid}}N\left(  \tau,1\right)  $ when a Bonferroni correction is
made for testing a total of $K$ hypotheses. \ For instance, with $\tau=0.1$
and $K=1$ hypothesis one needs $I=787$ pairs for 80\% power, but with $K=50$
hypotheses the Bonferroni correction raises that to 1,713 pairs for 80\%
power. \ In Table
%TCIMACRO{\TeXButton{reftabttest}{\ref{tabttest}}}%
%BeginExpansion
\ref{tabttest}%
%EndExpansion
, sample sizes that are more than twice the sample size for $K=1$ are in
\textbf{bold}.

In Table
%TCIMACRO{\TeXButton{reftabttest}{\ref{tabttest}}}%
%BeginExpansion
\ref{tabttest}%
%EndExpansion
, for $K=50$, 100, 250 and 500 hypotheses, it would be worth discarding a
random half of the data if that half correctly identified the one false null
hypothesis, thereby eliminated the need to correct for multiple testing in the
remaining half of the data. \ For $K=10$ hypotheses, this is not true, and use
of the Bonferroni correction with the full sample yields 80\% power with fewer pairs.

Table
%TCIMACRO{\TeXButton{reftabttest}{\ref{tabttest}} }%
%BeginExpansion
\ref{tabttest}
%EndExpansion
is a simple, informal guide, but it is an oversimplification in many ways.
\ First, if several hypotheses were false, not just one, use of the Bonferroni
adjustment might reject several hypotheses. \ Second, cross-screening performs
tests in both half-samples, but Table
%TCIMACRO{\TeXButton{reftabttest}{\ref{tabttest}} }%
%BeginExpansion
\ref{tabttest}
%EndExpansion
simply discards one half-sample. \ Third, the conventional t-test has size
equal to its level, so the issues in \S \ref{ssExpectedPvalue} and Table
%TCIMACRO{\TeXButton{reftabEPV}{\ref{tabEPV}} }%
%BeginExpansion
\ref{tabEPV}
%EndExpansion
that arise with $\Gamma>1$ are not reflected in Table
%TCIMACRO{\TeXButton{reftabttest}{\ref{tabttest}}}%
%BeginExpansion
\ref{tabttest}%
%EndExpansion
.

\subsection{Stylized asymptotics}

\label{ssAsymptotic}

One attraction of cross-screening is that it is highly flexible. \ For
instance, one half-sample can adaptively suggest the direction and the test
statistic to be used in the other half, and this might be done in a variety of
ways. \ In this section, we greatly restrict, and somewhat distort, the use of
cross-screening in the hope of gleaning some analytical insight into its
behavior relative to a test that does not split the sample in half. \ In
particular, we do not allow cross-screening to adaptively select a test
statistic, a major source of its power. \ Additionally, we assume that only
the first outcome is affected by the treatment and cross-screening selects
$K_{1}=K_{2}=1$ hypothesis to test in each half-sample.

Suppose the sample size, $I$, is even and the pairs are split at random in two
halves of size $I/2$. \ Let $\dot{T}$ and $\ddot{T}$ be the same test
statistic computed for the first outcome from the two halves of the data,
where we reject $H_{1}$ if the test statistic is large. \ Because there is
nothing adaptive here --- the choice of test statistic and the direction of
the test are fixed in advance --- the statistics $\dot{T}$ and $\ddot{T}$ are
independent and have the same distribution, whether under the null hypothesis
or the alternative. \ Suppose that the test statistic we would use if we did
not split the sample was $\dot{T}+\ddot{T}$. \ Suppose, finally, that $\dot
{T}$ and $\ddot{T}$ would each be $N\left(  \theta_{0},\upsilon_{0}\right)  $
under the null hypothesis and would each be $N\left(  \theta,\upsilon\right)
$ under the alternative. \ It is convenient to think of $\dot{T}$ and
$\ddot{T}$ as scaled so that, as $I\rightarrow\infty$, the expectations,
$\theta_{0}$ and $\theta$, are constant, while the variances, $\upsilon_{0}$
and $\upsilon$, are $O\left(  2/I\right)  $.

The situation just described would hold exactly with $\upsilon=\upsilon_{0} $
for the mean of $I/2$ iid Normally distributed pair differences, and it would
hold asymptotically for Wilcoxon's signed rank test and for many other tests
for paired data. \ The signed rank statistic for the whole sample is not quite
the same as the sum of the statistics for the two half samples, but as
$I\rightarrow\infty$ this distinction would become unimportant.
\ Additionally, $\upsilon\neq\upsilon_{0}$ for the signed rank statistic.

The combined test, $\dot{T}+\ddot{T}$ is corrected for multiple testing of $K$
hypotheses in $2$ tails, so let $\varrho_{b}=2K$. \ The combined test rejects
the null if $\left(  \dot{T}+\ddot{T}-2\theta_{0}\right)  /\sqrt{2\upsilon
_{0}}\geq-\Phi^{-1}\left(  \alpha/\varrho_{b}\right)  $, with power
\begin{equation}
\Pr\left\{  \frac{\dot{T}+\ddot{T}-2\theta}{\sqrt{2\upsilon}}\geq
\frac{2\left(  \theta_{0}-\theta\right)  -\Phi^{-1}\left(  \alpha/\varrho
_{b}\right)  \sqrt{2\upsilon_{0}}}{\sqrt{2\upsilon}}\right\}  =1-\Phi\left\{
\frac{\sqrt{2}\left(  \theta_{0}-\theta\right)  -\Phi^{-1}\left(
\alpha/\varrho_{b}\right)  \sqrt{\upsilon_{0}}}{\sqrt{\upsilon}}\right\}
\text{.} \label{eqAsympBonf}%
\end{equation}
Now, consider cross-screening. \ Suppose that the first half-sample recommends
testing hypothesis $k=1$ in the second half if $\left(  \dot{T}-\theta
_{0}\right)  /\sqrt{\upsilon_{0}}\geq-\Phi^{-1}\left(  \alpha\right)  $, and
the second half rejects the hypothesis if $\left(  \ddot{T}-\theta_{0}\right)
/\sqrt{\upsilon_{0}}\geq-\Phi^{-1}\left(  \alpha/\varrho_{s}\right)  $ where
$\varrho_{s}=2$ is the cross-screening correction testing for testing one
hypothesis in two half samples. \ In parallel, the second half recommends
testing if $\left(  \ddot{T}-\theta_{0}\right)  /\sqrt{\upsilon_{0}}\geq
-\Phi^{-1}\left(  \alpha\right)  $ and the first half rejects if $\left(
\dot{T}-\theta_{0}\right)  /\sqrt{\upsilon_{0}}\geq-\Phi^{-1}\left(
\alpha/\varrho_{s}\right)  $. \ So cross-screening rejects hypothesis $k=1$
if
\begin{equation}
\max\left(  \frac{\dot{T}-\theta}{\sqrt{\upsilon}},\,\frac{\ddot{T}-\theta
}{\sqrt{\upsilon}}\right)  \geq\frac{\theta_{0}-\theta-\Phi^{-1}\left(
\alpha/\varrho_{s}\right)  \sqrt{\upsilon_{0}}}{\sqrt{\upsilon}}=\lambda
_{1}\text{, say,} \label{eqAsympCross1}%
\end{equation}
and
\begin{equation}
\min\left(  \frac{\dot{T}-\theta}{\sqrt{\upsilon}},\,\frac{\ddot{T}-\theta
}{\sqrt{\upsilon}}\right)  \geq\frac{\theta_{0}-\theta-\Phi^{-1}\left(
\alpha\right)  \sqrt{\upsilon_{0}}}{\sqrt{\upsilon}}=\lambda_{2}\text{, say.}
\label{eqAympCross2}%
\end{equation}
\qquad

Actually, (\ref{eqAsympCross1}) and (\ref{eqAympCross2}) embody the following
small distortion. \ In principal, it is logically possible, but not probable,
that several of the $K$ hypotheses satisfy (\ref{eqAsympCross1}) and
(\ref{eqAympCross2}). \ This would complicate the situation; however, it is
improbable because only $H_{1}$ is false. \ For (\ref{eqAsympCross1}) and
(\ref{eqAympCross2}) to hold for a second hypothesis, we would need to reject
the same true null hypothesis twice, once in each half-sample. \

It is important to notice that $\sqrt{2}\left(  \theta_{0}-\theta\right)  $
appears on the right of (\ref{eqAsympBonf}), where $\theta_{0}-\theta$ appears
in (\ref{eqAsympCross1}) and (\ref{eqAympCross2}). \ Therefore, if
$I\rightarrow\infty$ with $\varrho_{b}$ and $\varrho_{s}$ fixed --- that is,
effectively with $K$ fixed --- the cross-screening method has little hope of
competing with the Bonferroni method. \ There is, however, a serious
competition if $I$ is fixed as $K$ increases, perhaps with $\varrho_{s}=2$
fixed and $\varrho_{b}=2K$ increasing.

With $\varrho_{s}=2$, we many calculate the probability of the joint event
(\ref{eqAsympCross1}) and (\ref{eqAympCross2}) as
\begin{equation}
2\int\nolimits_{\lambda_{1}}^{\infty}\int\nolimits_{\lambda_{2}}^{y}%
\phi\left(  x\right)  \,\phi\left(  y\right)  \,dx\,dy
\label{eqAsympCrossPower}%
\end{equation}
where $\phi\left(  \cdot\right)  $ is the standard Normal density.

In Table
%TCIMACRO{\TeXButton{reftabAsymp}{\ref{tabAsymp}}}%
%BeginExpansion
\ref{tabAsymp}%
%EndExpansion
, we compare the large-sample power of a nonadaptive cross-screening test of a
one-sided hypothesis with a Bonferroni adjusted test of $K$ two-sided
hypotheses, where the last $K-1$ null hypotheses are true. \ The family-wise
error rate is controlled at $\alpha=0.05$, so cross-screening tests in two
half samples at level $\alpha/2=0.025$, while the Bonferroni procedure does
$2K$ one-sided tests at level $\alpha/\left(  2K\right)  $. Here,
cross-screening is not permitted to use a major source of its power, namely
adaptive testing. \ Table
%TCIMACRO{\TeXButton{reftabAsymp}{\ref{tabAsymp}} }%
%BeginExpansion
\ref{tabAsymp}
%EndExpansion
assumes $\upsilon=\upsilon_{0}$, as for the Normal-mean situation above, and
characterizes power in terms of the noncentrality parameter $\mathrm{ncp}%
=\left(  \theta-\theta_{0}\right)  /\sqrt{\upsilon}$. \

In Table
%TCIMACRO{\TeXButton{reftabAsymp}{\ref{tabAsymp}}}%
%BeginExpansion
\ref{tabAsymp}%
%EndExpansion
, we see that cross-screening is a terrible way to select one tail of a two
tailed test if there is only $K=1$ hypothesis, that cross-screening is
inferior for $K=10$ hypotheses, but that it has higher power than the
Bonferroni procedure for $K=100$, 250 or 500 hypotheses. \ Our sense is that
Table
%TCIMACRO{\TeXButton{reftabAsymp}{\ref{tabAsymp}} }%
%BeginExpansion
\ref{tabAsymp}
%EndExpansion
offers correct qualitative advice: cross-screening is useful when searching
for a few large needles in a very big haystack, and is useless in searching
for a few hay-like needles in a small haystack. \ The simulation in
\S \ref{ssSimulation} will provide further numerical results about power free
of the small distortion noted above.

\subsection{Splitting with a small planning sample}

\label{ssSplitPlanning}

An alternative to cross-screening is single screening, as discussed in various
contexts by Cox (1975), Heller et al. (2009) and Zhang et al. (2011). \ In
single screening, the sample is split at random into two parts, a small
planning sample and a large analysis sample. \ The study is planned using the
planning sample, which is then discarded. \ For example, one might use the
planning sample to order the hypotheses $H_{k}$, determine the side of their
one-sided alternatives, and to select test statistics. \ Then, with the fixed
order, sides and test statistics, the analysis sample might test the
hypotheses in order. \ Importantly, split screening discards the planning
sample while cross-screening uses it in testing, but split screening plans
using a small planning sample while cross-screening uses a half sample.

It only takes a little thought to realize that single screening is better than
cross-screening as the sample size grows, $I\rightarrow\infty$, with the
number of hypotheses $K$ remaining fixed. \ See the cited papers for related
formal calculations. \ The reason is that the planning sample only needs to be
large enough to ensure that sensible plans are made, so that as $I\rightarrow
\infty$ the fraction of the sample needed for planning can diminish and the
loss of power from using a planning sample can diminish as well. \ In an
application with $I=132,786$ matched pairs, Zhang et al. (2011) used a 10\%
planning sample, where the 10\% reduction in sample size more than paid for
itself by improving the plan for analysis of the remaining 90\% of the sample;
moreover, $I/10\approx13,800$ pairs was an adequate sample for planning
purposes. \ It is doubtful that results for $I\rightarrow\infty$ with $K$
fixed are relevant to a situation like \S \ref{ssIntroExample} with $I=234$
pairs and $K=46$ hypotheses. \ Ten percent of $I=234$ pairs is only 23 pairs,
and that may be too small a planning sample to make correct decisions about
$K=46$ hypotheses.

The simulation in \S \ref{ssSimulation} compares cross-screening to split
screening with a 20\% planning sample.

\section{Simulation of a sensitivity analysis testing many null hypotheses}

\label{ssSimulation}

Table
%TCIMACRO{\TeXButton{reftabSim}{\ref{tabSim}} }%
%BeginExpansion
\ref{tabSim}
%EndExpansion
reports simulated power of a level-0.05 sensitivity analysis conducted with
$\Gamma=2$ with $K=100$ or $K=500$ hypotheses and $I=100$, or 250 or 500
matched pairs. \ Among the $K$ hypotheses, one or two are false. \ These are
situations in which cross-screening is expected to perform well: the number of
hypotheses is large compared with the sample size, most null hypotheses are
true, and a sensitivity analysis is performed with $\Gamma>1$. \ As noted
previously, cross-screening should not be used to test a small number of
hypotheses, $K$.

There are $K$ independent Normal outcomes with variance 1. \ Outcomes $k=1$
and $k=2$ have expectations $\tau_{1}$ and $\tau_{2}$, while $\tau_{k}=0$ for
outcomes $k=3,\ldots,K$. \ That is, hypothesis $H_{1}$ is always false with
$\tau_{1}\neq0$, $H_{2}$ is false whenever $\tau_{2}\neq0$, and $H_{k}$ is
true for $k=3,\ldots,K$. \ A value of $K$, $I$ and $\left(  \tau_{1},\tau
_{2}\right)  $ defines one sampling situation. \ In 10,000 replicates of each
sampling situation, Table
%TCIMACRO{\TeXButton{reftabSim}{\ref{tabSim}} }%
%BeginExpansion
\ref{tabSim}
%EndExpansion
reports the proportion of rejections of $H_{1}$, of $H_{2}$ and of both
$H_{1}$ and $H_{2}$. \ In each sampling situation, the highest power for each
hypothesis, $H_{1}$, $H_{2}$, $H_{12}$, is in bold.

There are three tests. \ The first is the familiar Wilcoxon signed rank test.
\ The second test is another signed rank statistic, namely the U-statistic
$\left(  8,5,8\right)  $ from Rosenbaum (2011). \ For many error
distributions, the U-statistic $\left(  8,5,8\right)  $ has a larger design
sensitivity than Wilcoxon's statistic. \ Both the Wilcoxon statistic and the
U-statistic $\left(  8,5,8\right)  $ are used as fixed tests. \ The third test
makes an adaptive choice between three U-statistics, namely $\left(
8,5,8\right)  $, $\left(  8,6,7\right)  $, and $\left(  8,7,8\right)  $. \ Of
these three, the U-statistic $\left(  8,7,8\right)  $ has the highest design
sensitivity for short-tailed distributions like the Normal distribution,
$\left(  8,6,7\right)  $ has the highest design sensitivity for long-tailed
distributions like the $t$-distribution with 3 degrees of freedom, and
$\left(  8,5,8\right)  $ is a compromise; see Rosenbaum (2011, Table 3).

Three methods are used to strongly control the family-wise error rate in
two-sided tests. \ The Bonferroni method splits 0.05 among various hypotheses
tests. \ For instance, using Wilcoxon's test, the Bonferroni method splits
0.05 equally among $K$ tests each with two tails, so a rejection occurs if the
one-sided $P$-value bound at $\Gamma=2$ is $\leq0.05/\left(  2K\right)  $.
\ The same approach is used with the U-statistic $\left(  8,5,8\right)  $.
\ For adaptive inference, the Bonferroni method uses all three tests, but
rejects if the smallest $P$-value bound (\ref{eqNormalApproxSen}) is
$\leq0.05/\left(  3\times2\times K\right)  $. \ The Bonferroni method attains
the largest design sensitivity of the three component tests; however, there
are better approaches to adaptive inference (Rosenbaum 2012).

Cross-screening splits the pairs in half at random, plans the analysis in the
first half, tests in the second half, then plans in the second half and tests
in the first half, with a two-fold Bonferroni correction for having done both
analyses. \ Cross-screening does one-tailed tests, not two-tailed tests,
having selected one tail based on the planning sample. \ For adaptive
inference, cross-screening picks one of the U-statistics $\left(
8,5,8\right)  $, $\left(  8,6,7\right)  $, and $\left(  8,7,8\right)  $ based
on the planning sample. \ Cross-screening orders the $K$ hypotheses based on
the planning sample, tests the hypotheses in the analysis sample in the given
order, and stops testing with the first acceptance. \ It is well known that
testing in order strongly controls the family-wise error rate; see, for
instance, Koch and Gansky (1996), Hsu and Berger (1999), Hommel and Kropf
(2005), and Rosenbaum (2008). \ The hypotheses were ordered by determining
their sensitivity to bias in the planning sample, measured by $\Gamma$ at
level $\alpha=0.05$, placing the least sensitive hypotheses first.

Single screening split the pairs at random into a 20\% planning sample and an
80\% analysis sample. \ It discards the 20\% planning sample and does one
analysis based on the 80\% analysis sample, without a correction for having
done two analyses. \ In contrast to cross-screening, split screening: (i)
omits the correction for having done two analyses, (ii) uses a larger 80\%
sample in its one analysis, rather than two 50\% samples, (iii) but discards
20\% of the data. \ Aside from these differences, the procedures for split
screening are the same as for cross-screening. \ Specifically, in
split-screening, the planning sample determines the tail of a one-tailed test,
the order for testing-in-order, and in adaptive inference it chooses the test statistic.

In Table
%TCIMACRO{\TeXButton{reftabSim}{\ref{tabSim}}}%
%BeginExpansion
\ref{tabSim}%
%EndExpansion
, the highest power in most sampling situations is from cross-screening with
an adaptive choice of test statistic. \ Wilcoxon's test has inferior power at
$\Gamma=2$, consistent with results in Rosenbaum (2011); however,
cross-screening often has higher power than the Bonferroni method when the
Wilcoxon test is used. \ The case of $\left(  \tau_{1},\tau_{2}\right)
=\left(  0.6,0.4\right)  $ for $I=500$ and $K=100$ is especially interesting.
\ Using adaptive methods, both the Bonferroni method and cross-screening
reject $H_{1}$ with power near 1, but cross-screening has much higher power
for $H_{2}$. \ A similar pattern is seen for $\left(  \tau_{1},\tau
_{2}\right)  =\left(  0.6,0.4\right)  $ for $I=250$ and $K=100$ and for
$I=500$ and $K=500$.

We repeated the simulation in Table
%TCIMACRO{\TeXButton{reftabSim}{\ref{tabSim}} }%
%BeginExpansion
\ref{tabSim}
%EndExpansion
but with samples from a t-distribution on 4 degrees of freedom rather than
from a Normal distribution. \ The comparison of Bonferroni, cross-screening
and single screening was similar to Table
%TCIMACRO{\TeXButton{reftabSim}{\ref{tabSim}}}%
%BeginExpansion
\ref{tabSim}%
%EndExpansion
, with cross-screening having superior power. \ Unlike Table
%TCIMACRO{\TeXButton{reftabSim}{\ref{tabSim}}}%
%BeginExpansion
\ref{tabSim}%
%EndExpansion
, adaptive cross-screening was slightly to cross-screening with a fixed choice
of the U-statistic $\left(  8,5,8\right)  $, but this is probably because
$\left(  8,5,8\right)  $ is an excellent choice for the t-distribution with 4
degrees of freedom; see Rosenbaum (2011, Table 4).

The simulation in Table
%TCIMACRO{\TeXButton{reftabSim}{\ref{tabSim}} }%
%BeginExpansion
\ref{tabSim}
%EndExpansion
used screening to order hypotheses which were then tested in a fixed order,
terminating with the first acceptance. \ We also simulated several alternative
methods. \ Wiens (2003), Hommel and Kropf (2005) and Burman et al. (2009)
proposed fixed sequence testing procedures that test-in-order at a level below
$\alpha$ so that they can continue testing beyond the first acceptance.
\ Wiens (2003) and Hommel and Kropf (2005) transfer forward unspent $\alpha$
to test later hypotheses in the sequence, whereas Burman et al. (2009) also
cycle back to retest early hypotheses with larger $\alpha$ when later
hypotheses are rejected. \ We tried a so-called fall-back procedure, testing
the first hypothesis in order at level $\alpha/2$, and the second hypothesis
at level $\alpha$ if the first hypothesis was rejected or at level $\alpha/2$
if the first hypothesis was not rejected. \ We also tried recycling, meaning
that if the first hypothesis was not rejected at level $\alpha/2$, but the
second hypothesis was rejected at level $\alpha/2$, then the first hypothesis
was retested at level $\alpha$. \ In the situations in Table
%TCIMACRO{\TeXButton{reftabSim}{\ref{tabSim}}}%
%BeginExpansion
\ref{tabSim}%
%EndExpansion
, there are at most two false null hypotheses, so recycling is logically
better than fall-back, but this need not be true when three null hypotheses
are false. \ Fixed sequence testing, as reported in Table
%TCIMACRO{\TeXButton{reftabSim}{\ref{tabSim}}}%
%BeginExpansion
\ref{tabSim}%
%EndExpansion
, fall-back and recycling were close competitors for each of the situations in
Table
%TCIMACRO{\TeXButton{reftabSim}{\ref{tabSim}}}%
%BeginExpansion
\ref{tabSim}%
%EndExpansion
, and there was no consistent winner for all situations. \ As logic would
suggest, fixed sequence testing had a slight advantage when only $H_{1}$ was
false because $\tau_{2}=0$. \ After all, fixed sequence testing bets all of
$\alpha$ on the first hypothesis in the sequence. \ Similarly, fall-back and
recycling had a slight advantage when $\left(  \tau_{1},\tau_{2}\right)
=\left(  0.5,0.5\right)  $. \

Although we evaluated ordered testing procedures, one need not use any form of
ordered testing to use cross-screening. \ Instead, one could use the first
sample to select $K_{1}\ll K$ hypotheses to test in the second sample, and use
the second to select $K_{2}\ll K$ hypotheses to test in the second sample, and
then correct for multiple testing using Holm's (1979) procedure. \ As always
with cross-screening, a two-fold Bonferroni correction is needed to analyze
both half samples.

\section{Nonrandom cross-screening}

\label{secNonrandomCS}

In some circumstances, it would be helpful to show that a treatment has the
same effect in each of two subpopulations. \ Nonrandom cross-screening uses an
observed binary covariate rather than random numbers to split the sample.
\ Nonrandom and random cross-screening each have advantages and disadvantages.

Suppose that a treatment has a large effect in one subpopulation and no effect
in the complementary subpopulation. \ In this case, nonrandom cross-screening
will be unhelpful, because each subpopulation will provide highly misleading
advice about how to analyze the complementary subpopulation.

Different people often receive the same treatment for different reasons. \ One
person eats fish because she lives on the coast of Maine and inexpensive fresh
fish is abundant, while another person eats fish in Arizona believing it to
confer health benefits. \ When treatments are not randomly assigned, the
evidence that the treatment is the cause of its ostensible effects is
strengthened by showing that people who receive the treatment for different
reasons experience similar effects; see Rosenbaum (2001; 2015b, \S 1.6) \ With
issues of this sort in mind, Lund and Bonaa (1993) examined the possible
effects of high fish consumption by comparing the wives of fisherman to
controls of similar socioeconomic status. \ Although the publicly available
NHANES data does not classify people by geography or employment, some other
data set might permit people who consume high levels of fish to be divided
based on whether they eat fish believing it to be health promoting or because
of its availability at low cost. \ In this case, the matched pairs might be
divided not at random but to distinguish two reasons people eat fish. \ The
study might be more convincing if it demonstrated the same or similar
ostensible effects in both types of pairs.

There is one key technical point about nonrandom cross-screening. \ If we view
a study as testing Fisher's null hypothesis that asserts the treatment has no
effect on anyone, then cross-screening may be used as above to test this
hypothesis while strongly controlling the family-wise error rate. \ However,
if both splits lead to rejection, then this achieves Bogomolov-Heller
replicability, thereby rejecting the null hypothesis of no effect for each
subpopulation. \ With random cross-screening, Bogomolov-Heller replicability
does not have a clear interpretation, but with nonrandom cross-screening it
constitutes a distinct strengthening of the study's conclusions.%

%TCIMACRO{\TeXButton{References}{\section*{References}}}%
%BeginExpansion
\section*{References}%
%EndExpansion
%

%TCIMACRO{\TeXButton{ref}{\setlength{\hangindent}{12pt}
%\noindent}}%
%BeginExpansion
\setlength{\hangindent}{12pt}
\noindent
%EndExpansion
Berk, R. H. and Jones, D. H. (1978), \textquotedblleft Relatively optimal
combinations of test statistics,\textquotedblright\ \textit{Scandinavian
Journal of Statistics}, 5, 158-162.%

%TCIMACRO{\TeXButton{ref}{\setlength{\hangindent}{12pt}
%\noindent}}%
%BeginExpansion
\setlength{\hangindent}{12pt}
\noindent
%EndExpansion
Bogomolov, M. and Heller, R. (2013), \textquotedblleft Discovering findings
that replicate from a primary study of high dimension to a follow-up
study,\textquotedblright\ \textit{Journal American Statistical Association},
108, 1480-1492.%

%TCIMACRO{\TeXButton{ref}{\setlength{\hangindent}{12pt}
%\noindent}}%
%BeginExpansion
\setlength{\hangindent}{12pt}
\noindent
%EndExpansion
Burman, C. F., Sonesson, C., and Guilbaud, O. (2009), \textquotedblleft A
recycling framework for the construction of Bonferroni-based multiple
tests,\textquotedblright\ \textit{Statistics in Medicine}, 28, 739-761.%

%TCIMACRO{\TeXButton{ref}{\setlength{\hangindent}{12pt}
%\noindent}}%
%BeginExpansion
\setlength{\hangindent}{12pt}
\noindent
%EndExpansion
Cox, D.R. (1975), \textquotedblleft A note on data-splitting for the
evaluation of significance levels,\textquotedblright\ \textit{Biometrika}, 62, 441-444.%

%TCIMACRO{\TeXButton{ref}{\setlength{\hangindent}{12pt}
%\noindent}}%
%BeginExpansion
\setlength{\hangindent}{12pt}
\noindent
%EndExpansion
Cox, D.R. (1977), \textquotedblleft The role of significance tests (with
Discussion),\textquotedblright\ \textit{Scandinavian Journal of Statistics}, 49-70.%

%TCIMACRO{\TeXButton{ref}{\setlength{\hangindent}{12pt}
%\noindent}}%
%BeginExpansion
\setlength{\hangindent}{12pt}
\noindent
%EndExpansion
Dempster, A. P. and Schatzoff, M. (1965), \textquotedblleft Expected
significance level as a sensitivity index for test
statistics,\textquotedblright\ \textit{Journal American Statistical
Association}, 60, 420-436.%

%TCIMACRO{\TeXButton{ref}{\setlength{\hangindent}{12pt}
%\noindent}}%
%BeginExpansion
\setlength{\hangindent}{12pt}
\noindent
%EndExpansion
Fisher, R. A. (1935), \textit{Design of Experiments}, Ediburgh: Oliver and Boyd.%

%TCIMACRO{\TeXButton{ref}{\setlength{\hangindent}{12pt}
%\noindent}}%
%BeginExpansion
\setlength{\hangindent}{12pt}
\noindent
%EndExpansion
Fogarty, C. B. and Small, D. S. (2016), \textquotedblleft Sensitivity analysis
for multiple comparisons in matched observational studies through
quadratically constrained linear programming,\textquotedblright%
\ \textit{Journal of the American Statistical Association}, 111, 1820-1830.%

%TCIMACRO{\TeXButton{ref}{\setlength{\hangindent}{12pt}
%\noindent}}%
%BeginExpansion
\setlength{\hangindent}{12pt}
\noindent
%EndExpansion
Heller, R., Rosenbaum, P.R. and Small, D.S. (2009), \textquotedblleft Split
samples and design sensitivity in observational studies,\textquotedblright%
\ \textit{Journal of the American Statistical Association}, 104, 1090-1101.%

%TCIMACRO{\TeXButton{ref}{\setlength{\hangindent}{12pt}
%\noindent}}%
%BeginExpansion
\setlength{\hangindent}{12pt}
\noindent
%EndExpansion
Holm, S. (1979), \textquotedblleft A simple sequentially rejective multiple
test procedure,\textquotedblright\ \textit{Scandinavian Journal of
Statistics},\textquotedblright\ 6, 65-70.%

%TCIMACRO{\TeXButton{ref}{\setlength{\hangindent}{12pt}
%\noindent}}%
%BeginExpansion
\setlength{\hangindent}{12pt}
\noindent
%EndExpansion
Hommel, G., and Kropf, S. (2005), \textquotedblleft Tests for differentiation
in gene expression using a data-driven order or weights for
hypotheses,\textquotedblright\ \textit{Biometrical Journal}, 47, 554-562.%

%TCIMACRO{\TeXButton{ref}{\setlength{\hangindent}{12pt}
%\noindent}}%
%BeginExpansion
\setlength{\hangindent}{12pt}
\noindent
%EndExpansion
Hsu, J. C. and Berger, R. L. (1999), \textquotedblleft Stepwise confidence
intervals without multiplicity adjustment for dose---response and toxicity
studies,\textquotedblright\ \textit{Journal of the American Statistical
Association}, 94, 468-482.%

%TCIMACRO{\TeXButton{ref}{\setlength{\hangindent}{12pt}
%\noindent}}%
%BeginExpansion
\setlength{\hangindent}{12pt}
\noindent
%EndExpansion
Koch, G. G. and Gansky, S. A. (1996), \textquotedblleft Statistical
considerations for multiplicity in confirmatory protocols,\textquotedblright%
\ \textit{Drug Information Journal}, 30(2), 523-534.%

%TCIMACRO{\TeXButton{ref}{\setlength{\hangindent}{12pt}
%\noindent}}%
%BeginExpansion
\setlength{\hangindent}{12pt}
\noindent
%EndExpansion
Lund, E. and Bonaa, K. H. (1993), \textquotedblleft Reduced breast cancer
mortality among fisherman's wives in Norway,\textquotedblright\ \textit{Cancer
Causes and Control}, 4, 283-287.%

%TCIMACRO{\TeXButton{ref}{\setlength{\hangindent}{12pt}
%\noindent}}%
%BeginExpansion
\setlength{\hangindent}{12pt}
\noindent
%EndExpansion
Neyman, J. (1923, 1990), \textquotedblleft On the application of probability
theory to agricultural experiments,\textquotedblright\ reprinted in English
in: \textit{Statistical Science}, 5, 463-480.%

%TCIMACRO{\TeXButton{ref}{\setlength{\hangindent}{12pt}
%\noindent}}%
%BeginExpansion
\setlength{\hangindent}{12pt}
\noindent
%EndExpansion
Rosenbaum, P.R. (1987), \textquotedblleft Sensitivity analysis for certain
permutation inferences in matched observational studies,\textquotedblright%
\ \textit{Biometrika}, 74, 13-26.%

%TCIMACRO{\TeXButton{ref}{\setlength{\hangindent}{12pt}
%\noindent}}%
%BeginExpansion
\setlength{\hangindent}{12pt}
\noindent
%EndExpansion
Rosenbaum, P. R. (2001), \textquotedblleft Replicating effects and
biases,\textquotedblright\ \textit{American Statistician}, 55, 223--227.%

%TCIMACRO{\TeXButton{ref}{\setlength{\hangindent}{12pt}
%\noindent}}%
%BeginExpansion
\setlength{\hangindent}{12pt}
\noindent
%EndExpansion
Rosenbaum, P. R. (2002), \textit{Observational Studies} (2$^{nd}$ edition),
New York: Springer.%

%TCIMACRO{\TeXButton{ref}{\setlength{\hangindent}{12pt}
%\noindent}}%
%BeginExpansion
\setlength{\hangindent}{12pt}
\noindent
%EndExpansion
Rosenbaum, P. R. (2004), \textquotedblleft Design sensitivity in observational
studies,\textquotedblright\ \textit{Biometrika}, 91, 153-164.%

%TCIMACRO{\TeXButton{ref}{\setlength{\hangindent}{12pt}
%\noindent}}%
%BeginExpansion
\setlength{\hangindent}{12pt}
\noindent
%EndExpansion
Rosenbaum, P. R. (2008), \textquotedblleft Testing hypotheses in
order,\textquotedblright\ \textit{Biometrika}, 95, 248-252.%

%TCIMACRO{\TeXButton{ref}{\setlength{\hangindent}{12pt}
%\noindent}}%
%BeginExpansion
\setlength{\hangindent}{12pt}
\noindent
%EndExpansion
Rosenbaum, P. R. (2010), \textit{Design of Observational Studies}, New York: Springer.%

%TCIMACRO{\TeXButton{ref}{\setlength{\hangindent}{12pt}
%\noindent}}%
%BeginExpansion
\setlength{\hangindent}{12pt}
\noindent
%EndExpansion
Rosenbaum, P. R. (2011), \textquotedblleft A New u-Statistic with superior
design sensitivity in matched observational studies,\textquotedblright%
\ \textit{Biometrics}, 67, 1017-1027.%

%TCIMACRO{\TeXButton{ref}{\setlength{\hangindent}{12pt}
%\noindent}}%
%BeginExpansion
\setlength{\hangindent}{12pt}
\noindent
%EndExpansion
Rosenbaum, P. R. (2012), \textquotedblleft Testing One Hypothesis Twice in
Observational Studies,\textquotedblright\ \textit{Biometrika}, 99, 763--774.%

%TCIMACRO{\TeXButton{ref}{\setlength{\hangindent}{12pt}
%\noindent}}%
%BeginExpansion
\setlength{\hangindent}{12pt}
\noindent
%EndExpansion
Rosenbaum, P. R. (2015a), \textquotedblleft Bahadur efficiency of sensitivity
analyses in observational studies,\textquotedblright\ \textit{Journal of the
American Statistical Association}, 110, 205-217.%

%TCIMACRO{\TeXButton{ref}{\setlength{\hangindent}{12pt}
%\noindent}}%
%BeginExpansion
\setlength{\hangindent}{12pt}
\noindent
%EndExpansion
Rosenbaum, P. R. (2015b), \textquotedblleft How to see more in observational
studies: Some new quasi-experimental devices,\textquotedblright%
\ \textit{Annual Review of Statistics and Its Application}, 2, 21-48.%

%TCIMACRO{\TeXButton{ref}{\setlength{\hangindent}{12pt}
%\noindent}}%
%BeginExpansion
\setlength{\hangindent}{12pt}
\noindent
%EndExpansion
Rosenbaum, P. R. (2016), \textquotedblleft Using Scheff\'{e} projections for
multiple outcomes in an observational study of smoking and periodontal
disease,\textquotedblright\ \textit{Annals of Applied Statistics}, 10, 1447-1471.%

%TCIMACRO{\TeXButton{ref}{\setlength{\hangindent}{12pt}
%\noindent}}%
%BeginExpansion
\setlength{\hangindent}{12pt}
\noindent
%EndExpansion
Rosenbaum, P.R. and Silber, J.H. (2009a), \textquotedblleft Sensitivity
analysis for equivalence and difference in an observational study of neonatal
intensive care units,\textquotedblright\ \textit{Journal of the American
Statistical Association}, 104, 501-511.%

%TCIMACRO{\TeXButton{ref}{\setlength{\hangindent}{12pt}
%\noindent}}%
%BeginExpansion
\setlength{\hangindent}{12pt}
\noindent
%EndExpansion
Rosenbaum, P. R. and Silber, J. H. (2009b), \textquotedblleft Amplification of
sensitivity analysis in observational studies,\textquotedblright%
\ \textit{Journal American Statistical Association,} 104, 1398-1405.
\ (\texttt{amplify} function in the \texttt{R} package \texttt{sensitivitymv})%

%TCIMACRO{\TeXButton{ref}{\setlength{\hangindent}{12pt}
%\noindent}}%
%BeginExpansion
\setlength{\hangindent}{12pt}
\noindent
%EndExpansion
Rubin, D. B. (1974), \textquotedblleft Estimating causal effects of treatments
in randomized and nonrandomized studies,\textquotedblright\ \textit{Journal of
Educational Psychology}, 66, 688-701.%

%TCIMACRO{\TeXButton{ref}{\setlength{\hangindent}{12pt}
%\noindent}}%
%BeginExpansion
\setlength{\hangindent}{12pt}
\noindent
%EndExpansion
Sackrowitz, H. and Samuel-Cahn, E. (1999), \textquotedblleft$P$-values as
random variables --- expected $P$-values,\textquotedblright\ \textit{American
Statistician}, 53, 326-331.%

%TCIMACRO{\TeXButton{ref}{\setlength{\hangindent}{12pt}
%\noindent}}%
%BeginExpansion
\setlength{\hangindent}{12pt}
\noindent
%EndExpansion
Shaffer, J. P. (1974), \textquotedblleft Bidirectional unbiased
procedures,\textquotedblright\ \textit{Journal of the American Statistical
Association}, 69, 437-439.%

%TCIMACRO{\TeXButton{ref}{\setlength{\hangindent}{12pt}
%\noindent}}%
%BeginExpansion
\setlength{\hangindent}{12pt}
\noindent
%EndExpansion
Stephenson, W. R. (1981), \textquotedblleft A general class of one-sample
nonparametric test statistics based on subsamples,\textquotedblright%
\ \textit{Journal of the American Statistical Association}, 76, 960--966.%

%TCIMACRO{\TeXButton{ref}{\setlength{\hangindent}{12pt}
%\noindent}}%
%BeginExpansion
\setlength{\hangindent}{12pt}
\noindent
%EndExpansion
Wiens, B. L. (2003), \textquotedblleft A fixed sequence Bonferroni procedure
for testing multiple endpoints,\textquotedblright\ \textit{Pharmaceutical
Statistics}, 2, 211-215.%

%TCIMACRO{\TeXButton{ref}{\setlength{\hangindent}{12pt}
%\noindent}}%
%BeginExpansion
\setlength{\hangindent}{12pt}
\noindent
%EndExpansion
Zhang, K., Small, D. S., Lorch, S., Srinivas, S. and Rosenbaum, P. R. (2011),
\textquotedblleft Using split samples and evidence factors in an observational
study of neonatal outcomes,\textquotedblright\ \textit{Journal of the American
Statistical Association}, 106, 511-524.%

%TCIMACRO{\TeXButton{ref}{\setlength{\hangindent}{12pt}
%\noindent}}%
%BeginExpansion
\setlength{\hangindent}{12pt}
\noindent
%EndExpansion
Zubizarreta, J. R., Cerd\'{a}, M. and Rosenbaum, P. R. (2013),
\textquotedblleft Effect of the 2010 Chilean earthquake on posttraumatic
stress: reducing sensitivity to unmeasured bias through study
design,\textquotedblright\ \textit{Epidemiology}, 24, 79-87.%

\begin{table}[ht]
\caption{ Bonferroni correction (B) versus cross-screening (CS) with
$K=46$ outcomes and $\Gamma
=1$, 1.25, 9, and 11.  A blank indicates that no test was performed.}%
\label{tabCompare}
\small
\renewcommand{\arraystretch}{1}% for the vertical padding
\setlength{\tabcolsep}{0.4em}
\centering\begin{tabular}{| llc | cc | cc | cc | cc |}
\hline\multicolumn{3}{|c}{Sensitivity Parameter} & \multicolumn{2}{|c}%
{$\Gamma=1$} & \multicolumn{2}{|c}{$\Gamma=1.25$} & \multicolumn{2}%
{|c}{$\Gamma=9$}
& \multicolumn{2}{|c|}{$\Gamma=11$} \\ \hline\multicolumn{2}{|c}{Outcome}
& $k$ & B & CS & B & CS & B & CS & B & CS \\
\hline LBXSAL &  Albumin & 1 & 1.000 &  & 1.000 &  & 1.000 &  & 1.000 &  \\
LBXSBU &  Blood urea nitrogen & 2 & 1.000 &  & 1.000 &  & 1.000 &  & 1.000 &  \\
LBXSCA &  Total calcium & 3 & 1.000 &  & 1.000 &  & 1.000 &  & 1.000 &  \\
LBXSCH &  Cholesterol & 4 & 1.000 &  & 1.000 &  & 1.000 &  & 1.000 &  \\
LBXSCK &  Creatine phospho. & 5 & 1.000 &  & 1.000 &  & 1.000 &  & 1.000 &  \\
LBXSCR &  Creatinine & 6 & 1.000 &  & 1.000 &  & 1.000 &  & 1.000 &  \\
LBXSGB &  Globulin & 7 & 1.000 &  & 1.000 &  & 1.000 &  & 1.000 &  \\
LBXSGL &  Glucose & 8 & 1.000 &  & 1.000 &  & 1.000 &  & 1.000 &  \\
LBXSIR &  Iron & 9 & 1.000 &  & 1.000 &  & 1.000 &  & 1.000 &  \\
LBXSPH &  Phosphorus & 10 & 1.000 &  & 1.000 &  & 1.000 &  & 1.000 &  \\
LBXSTB &  Total bilirubin & 11 & 1.000 &  & 1.000 &  & 1.000 &  & 1.000 &  \\
LBXSTP &  Total protein & 12 & 1.000 &  & 1.000 &  & 1.000 &  & 1.000 &  \\
LBXSTR &  Triglycerides & 13 & 1.000 &  & 1.000 &  & 1.000 &  & 1.000 &  \\
LBXSUA &  Uric acid & 14 & 1.000 &  & 1.000 &  & 1.000 &  & 1.000 &  \\
WTSH2YR &  Blood metals & 15 & 0.024 &  & 1.000 &  & 1.000 &  & 1.000 &  \\
LBXBPB &  Blood lead & 16 & 1.000 &  & 1.000 &  & 1.000 &  & 1.000 &  \\
LBXBCD &  Blood cadmium & 17 & 1.000 &  & 1.000 &  & 1.000 &  & 1.000 &  \\
LBXTHG &  Blood mercury & 18 & 0.000 & 0.000 & 0.000 & 0.000 & 0.095 & 0.015 & 0.505 & 0.035 \\
LBXBSE &  Blood selenium & 19 & 0.380 &  & 1.000 &  & 1.000 &  & 1.000 &  \\
LBXBMN &  Blood manganese & 20 & 1.000 &  & 1.000 &  & 1.000 &  & 1.000 &  \\
LBXIHG &  Mercury, inorganic & 21 & 0.000 &  & 0.000 &  & 1.000 &  & 1.000 &  \\
LBXBGE &  Mercury, ethyl & 22 & 1.000 &  & 1.000 &  & 1.000 &  & 1.000 &  \\
LBXBGM &  Mercury, methyl & 23 & 0.000 & 0.000 & 0.000 & 0.000 & 0.075 & 0.014 & 0.405 & 0.031 \\
LBDHDD &  HDL-Cholesterol  & 24 & 1.000 &  & 1.000 &  & 1.000 &  & 1.000 &  \\
LBXWBCSI &  White blood cell cnt. & 25 & 1.000 &  & 1.000 &  & 1.000 &  & 1.000 &  \\
LBXLYPCT &  Lymphocyte \% & 26 & 1.000 &  & 1.000 &  & 1.000 &  & 1.000 &  \\
LBXMOPCT &  Monocyte \% & 27 & 1.000 &  & 1.000 &  & 1.000 &  & 1.000 &  \\
LBXNEPCT &  Seg. neutrophils \% & 28 & 1.000 &  & 1.000 &  & 1.000 &  & 1.000 &  \\
LBXEOPCT &  Eosinophils \% & 29 & 1.000 & & 1.000 &  & 1.000 &  & 1.000 &  \\
LBXBAPCT &  Basophils \% & 30 & 1.000 &  & 1.000 &  & 1.000 &  & 1.000 &  \\
LBDLYMNO &  Lymphocyte \# & 31 & 1.000 &  & 1.000 &  & 1.000 &  & 1.000 &  \\
LBDMONO &  Monocyte \# & 32 & 1.000 &  & 1.000 &  & 1.000 &  & 1.000 &  \\
LBDNENO &  Seg. neutrophils \# & 33 & 1.000 &  & 1.000 &  & 1.000 &  & 1.000 &  \\
LBDEONO &  Eosinophils \# & 34 & 1.000 &  & 1.000 &  & 1.000 &  & 1.000 &  \\
LBDBANO &  Basophils \# & 35 & 1.000 &  & 1.000 &  & 1.000 &  & 1.000 &  \\
LBXRBCSI &  Red blood cell cnt. & 36 & 1.000 &  & 1.000 &  & 1.000 &  & 1.000 &  \\
LBXHGB &  Hemoglobin & 37 & 1.000 &  & 1.000 &  & 1.000 &  & 1.000 &  \\
LBXHCT &  Hematocrit & 38 & 1.000 &  & 1.000 &  & 1.000 &  & 1.000 &  \\
LBXMCVSI &  Mean cell volume & 39 & 1.000 &  & 1.000 &  & 1.000 &  & 1.000 &  \\
LBXMCHSI &  Mean cell hemoglobin & 40 & 1.000 &  & 1.000 &  & 1.000 &  & 1.000 &  \\
LBXMC &  MCHC & 41 & 1.000 &  & 1.000 &  & 1.000 &  & 1.000 &  \\
LBXRDW &  Red cell dist. & 42 & 0.520 &  & 1.000 &  & 1.000 &  & 1.000 &  \\
LBXPLTSI &  Platelet count & 43 & 1.000 &  & 1.000 &  & 1.000 &  & 1.000 &  \\
LBXMPSI &  Mean platelet vol.  & 44 & 1.000 &  & 1.000 &  & 1.000 &  & 1.000 &  \\
LBXGH &  Glycohemoglobin  & 45 & 1.000 &  & 1.000 &  & 1.000 &  & 1.000 &  \\
BPXSY &  Systolic BP & 46 & 0.523 &  & 1.000 &  & 1.000 &  & 1.000 &  \\
\hline\end{tabular}
\end{table}%
%EndExpansion
%

%TCIMACRO{\TeXButton{tabEPV}{\begin{table}[ht]
%\caption{ When $H_{k}%
%$ is true and the bias in treatment assignment is at most $\Gamma'$,
%the table gives the upper bound on the size of an $\alpha
%=0.05$ level Wilcoxon test
%and the upper bound on the expected $P$-value (EPV) if
%the sensitivity analysis is performed with sensitivity parameter $\Gamma
%\ge\Gamma'$}.\label{tabEPV}
%\centering\begin{tabular}{ c c | c c c | c c c}
%\hline& & \multicolumn{3}{|c}{ Bound on Size } & \multicolumn{3}%
%{|c}{ Bound on EPV} \\ \hline$\Gamma'$ & $\Gamma
%$ & $I=100$ & $I=250$ & $I=500$ & $I=100$ & $I=250$ & $I=500$ \\ \hline
%1.00 & 1.00 & 0.05000 & 0.05000 & 0.05000 & 0.50 & 0.50 & 0.50 \\
%1.00 & 1.10 & 0.01976 & 0.01077 & 0.00511 & 0.62 & 0.68 & 0.74 \\
%1.00 & 1.25 & 0.00445 & 0.00074 & 0.00007 & 0.75 & 0.86 & 0.94 \\
%1.10 & 1.25 & 0.01392 & 0.00586 & 0.00197 & 0.65 & 0.73 & 0.81 \\
%1.25  & 1.25 & 0.05000 & 0.05000 & 0.05000  & 0.50 &  0.50 &  0.50 \\
%1.25 & 1.50 & 0.00750 & 0.00194 & 0.00033 & 0.71 & 0.81 & 0.89 \\
%1.50 & 2.00 & 0.00204 & 0.00017 & 0.00001 & 0.80 & 0.91 & 0.97 \\
%\hline\end{tabular}
%\end{table}}}%
%BeginExpansion
\begin{table}[ht]
\caption{ When $H_{k}%
$ is true and the bias in treatment assignment is at most $\Gamma'$,
the table gives the upper bound on the size of an $\alpha
=0.05$ level Wilcoxon test
and the upper bound on the expected $P$-value (EPV) if
the sensitivity analysis is performed with sensitivity parameter $\Gamma
\ge\Gamma'$}.\label{tabEPV}
\centering\begin{tabular}{ c c | c c c | c c c}
\hline& & \multicolumn{3}{|c}{ Bound on Size } & \multicolumn{3}%
{|c}{ Bound on EPV} \\ \hline$\Gamma'$ & $\Gamma
$ & $I=100$ & $I=250$ & $I=500$ & $I=100$ & $I=250$ & $I=500$ \\ \hline
1.00 & 1.00 & 0.05000 & 0.05000 & 0.05000 & 0.50 & 0.50 & 0.50 \\
1.00 & 1.10 & 0.01976 & 0.01077 & 0.00511 & 0.62 & 0.68 & 0.74 \\
1.00 & 1.25 & 0.00445 & 0.00074 & 0.00007 & 0.75 & 0.86 & 0.94 \\
1.10 & 1.25 & 0.01392 & 0.00586 & 0.00197 & 0.65 & 0.73 & 0.81 \\
1.25  & 1.25 & 0.05000 & 0.05000 & 0.05000  & 0.50 &  0.50 &  0.50 \\
1.25 & 1.50 & 0.00750 & 0.00194 & 0.00033 & 0.71 & 0.81 & 0.89 \\
1.50 & 2.00 & 0.00204 & 0.00017 & 0.00001 & 0.80 & 0.91 & 0.97 \\
\hline\end{tabular}
\end{table}%
%EndExpansion
%

%TCIMACRO{\TeXButton{tabttest}{\begin{table}[ht]
%\caption
%{Number of pairs required for power 0.8 with family-wise error rate 0.05 using the
%two-sided t-test with a Bonferroni correction for performing $K$ hypothesis tests.
%When the value in column $K$ is more than twice the value in column $K=1$,
%the sample size is in \textbf{bold}.}
%\label{tabttest}
%\centering\begin{tabular}{c|rrrrrr}
%\hline$\tau$  & K=1 & K=10 & K=50 & K=100 & K=250 & K=500 \\ \hline
%.1 &  787 & 1335 & \textbf{1713} & \textbf{1874} & \textbf{2087}
%& \textbf{2247} \\
%.3 &  89 & 152 & \textbf{195} & \textbf{214} & \textbf{238} & \textbf{256} \\
%.5 &  33 & 57 & \textbf{74} & \textbf{81} & \textbf{90} & \textbf{97} \\
%\hline\end{tabular}
%\end{table}}}%
%BeginExpansion
\begin{table}[ht]
\caption
{Number of pairs required for power 0.8 with family-wise error rate 0.05 using the
two-sided t-test with a Bonferroni correction for performing $K$ hypothesis tests.
When the value in column $K$ is more than twice the value in column $K=1$,
the sample size is in \textbf{bold}.}
\label{tabttest}
\centering\begin{tabular}{c|rrrrrr}
\hline$\tau$  & K=1 & K=10 & K=50 & K=100 & K=250 & K=500 \\ \hline
.1 &  787 & 1335 & \textbf{1713} & \textbf{1874} & \textbf{2087}
& \textbf{2247} \\
.3 &  89 & 152 & \textbf{195} & \textbf{214} & \textbf{238} & \textbf{256} \\
.5 &  33 & 57 & \textbf{74} & \textbf{81} & \textbf{90} & \textbf{97} \\
\hline\end{tabular}
\end{table}%
%EndExpansion
%

%TCIMACRO{\TeXButton{tabAsymp}{\begin{table}[ht]
%\caption
%{Comparison of the large-sample power of 0.05-level testing of one, one-sided hypothesis
%by cross-screening (CS)
%versus testing $K=1$, 10, 50, 100, 250 and 500 two-sided hypotheses using
%the Bonferroni method.
%Cases in which the Bonferroni method has lower power than CS are in \textbf
%{bold}.}
%\label{tabAsymp}
%\centering\begin{tabular}{r|r|rrrrrr}
%\hline& & \multicolumn{6}{|c}{ Bonferroni with $K$ hypotheses} \\ \hline
%ncp &  CS & 1 & 10 & 50 & 100 & 250 & 500 \\  \hline
%1 & 0.0591 & 0.2926 & 0.0818 & \textbf{0.0303} & \textbf{0.0194}
%& \textbf{0.0106} & \textbf{0.0066} \\
%2 & 0.3929 & 0.8074 & 0.5085 & \textbf{0.3220} & \textbf{0.2571}
%& \textbf{0.1866} & \textbf{0.1441} \\
%3 & 0.8285 & 0.9888 & 0.9244 & 0.8295 & \textbf{0.7769} & \textbf{0.6997}
%& \textbf{0.6376} \\
%\hline\end{tabular}
%\end{table}}}%
%BeginExpansion
\begin{table}[ht]
\caption
{Comparison of the large-sample power of 0.05-level testing of one, one-sided hypothesis
by cross-screening (CS)
versus testing $K=1$, 10, 50, 100, 250 and 500 two-sided hypotheses using
the Bonferroni method.
Cases in which the Bonferroni method has lower power than CS are in \textbf
{bold}.}
\label{tabAsymp}
\centering\begin{tabular}{r|r|rrrrrr}
\hline& & \multicolumn{6}{|c}{ Bonferroni with $K$ hypotheses} \\ \hline
ncp &  CS & 1 & 10 & 50 & 100 & 250 & 500 \\  \hline
1 & 0.0591 & 0.2926 & 0.0818 & \textbf{0.0303} & \textbf{0.0194}
& \textbf{0.0106} & \textbf{0.0066} \\
2 & 0.3929 & 0.8074 & 0.5085 & \textbf{0.3220} & \textbf{0.2571}
& \textbf{0.1866} & \textbf{0.1441} \\
3 & 0.8285 & 0.9888 & 0.9244 & 0.8295 & \textbf{0.7769} & \textbf{0.6997}
& \textbf{0.6376} \\
\hline\end{tabular}
\end{table}%
\begin{table}
\caption{Simulated power at $\Gamma
=2$ with Normal errors for $K$ independent outcomes in $I$ matched pairs.
For outcome $k=1$, hypotheses $H_{1}%
$ is false with pair differences symmetric about $\tau_{1} \ne0$.
For outcome $k=2$, hypotheses $H_{2}$ is false whenever $\tau_{2} \ne0$.
Null hypotheses $k=3,\dots
,K$ are true, with pair differences that are symmetric about zero.
Tabulated values are the proportion of rejections at the 0.05-level in 10,000 replicates.  The column
labeled $H_{12}$ is the proportion of times both $H_{1}$ and $H_{2}%
$ were both rejected. In each
sampling situation, the highest power for each hypothesis is in \textbf{bold}.
}\label{tabSim}
\renewcommand{\arraystretch}{1}% for the vertical padding
\setlength{\tabcolsep}{0.4em}
\small
\begin{tabular}{cccl|rrr|rrr|rrrrrrrrrrrrr}
\hline& & & & \multicolumn{3}{c|}{Bonferroni} & \multicolumn{3}{c|}%
{Cross screening} & \multicolumn{3}{c}{Single screening} \\
$K$ & $I$ & $(\tau_{1},\tau_{2})$ & Statistic & $H_1$ & $H_2$ & $H_{12}%
$ & $H_1$ & $H_2$ & $H_{12}$ & $H_1$ &
$H_2$ & $H_{12}$ \\
\hline
100 & 100 & (0.5,0.0) & Wilcoxon  & $1.4$ & $$ & $$ & $23.0$ & $$ & $$ & $11.2$ & $$ & $$ \\
&  &  & (8,5,8)  & $2.9$ & $$ & $$ & $\mathbf{28.4}%
$ & $$ & $$ & $13.2$ & $$ & $$ \\
&  &  & Adaptive  & $1.3$ & $$ & $$ & $27.4$ & $$ & $$ & $11.2$ & $$ & $$ \\ \hline
&  & (0.5,0.5) & Wilcoxon  & $1.2$ & $1.5$ & $0.0$ & $15.2$ & $16.0$ & $3.6$ & $10.3$ & $10.8$ & $1.8$ \\
&  &  & (8,5,8)  & $3.0$ & $2.8$ & $0.1$ & $19.1$ & $19.6$ & $5.4$ & $11.8$ & $11.5$ & $2.0$ \\
&  &  & Adaptive  & $1.2$ & $1.0$ & $0.0$ & $\mathbf{21.3}$ & $\mathbf
{21.6}$ & $\mathbf{6.7}$ & $10.9$ & $11.4$ & $1.8$ \\ \hline
&  & (0.6,0.4) & Wilcoxon  & $9.6$ & $0.2$ & $0.0$ & $47.8$ & $3.7$ & $2.6$ & $28.8$ & $2.5$ & $1.1$ \\
&  &  & (8,5,8)  & $14.5$ & $0.2$ & $0.0$ & $\mathbf{52.1}%
$ & $5.0$ & $3.5$ & $28.5$ & $3.0$ & $1.1$ \\
&  &  & Adaptive  & $7.8$ & $0.1$ & $0.0$ & $50.9$ & $\mathbf{6.1}%
$ & $\mathbf{4.2}$ & $25.5$ & $3.7$ & $1.3$ \\ \hline
& 250 & (0.5,0.0) & Wilcoxon  & $18.7$ & $$ & $$ & $72.7$ & $$ & $$ & $57.8$ & $$ & $$ \\
&  &  & (8,5,8)  & $31.1$ & $$ & $$ & $82.1$ & $$ & $$ & $60.9$ & $$ & $$ \\
&  &  & Adaptive  & $35.0$ & $$ & $$ & $\mathbf{88.3}%
$ & $$ & $$ & $60.0$ & $$ & $$ \\ \hline
&  & (0.5,0.5) & Wilcoxon  & $19.4$ & $17.9$ & $3.5$ & $53.7$ & $53.1$ & $39.8$ & $54.0$ & $53.3$ & $34.0$ \\
&  &  & (8,5,8)  & $30.5$ & $30.5$ & $9.5$ & $66.7$ & $67.1$ & $56.3$ & $59.0$ & $59.2$ & $39.1$ \\
&  &  & Adaptive  & $36.3$ & $35.1$ & $12.6$ & $\mathbf{79.3}$ & $\mathbf
{79.2}$ & $\mathbf{70.7}$ & $59.5$ & $59.3$ & $38.9$ \\ \hline
&  & (0.6,0.4) & Wilcoxon  & $71.4$ & $0.9$ & $0.6$ & $90.4$ & $20.7$ & $20.6$ & $78.5$ & $15.3$ & $13.7$ \\
&  &  & (8,5,8)  & $83.0$ & $2.8$ & $2.3$ & $93.6$ & $31.8$ & $31.7$ & $79.9$ & $21.4$ & $19.3$ \\
&  &  & Adaptive  & $84.2$ & $4.2$ & $3.5$ & $\mathbf{95.5}$ & $\mathbf
{44.4}$ & $\mathbf{44.2}$ & $78.4$ & $25.5$ & $22.2$ \\ \hline
& 500 & (0.5,0.0) & Wilcoxon  & $66.7$ & $$ & $$ & $96.3$ & $$ & $$ & $95.3$ & $$ & $$ \\
&  &  & (8,5,8)  & $82.8$ & $$ & $$ & $98.8$ & $$ & $$ & $95.9$ & $$ & $$ \\
&  &  & Adaptive  & $92.7$ & $$ & $$ & $\mathbf{99.8}%
$ & $$ & $$ & $95.0$ & $$ & $$ \\ \hline
&  & (0.5,0.5) & Wilcoxon  & $66.8$ & $66.5$ & $44.6$ & $91.3$ & $91.1$ & $88.2$ & $93.8$ & $93.9$ & $90.1$ \\
&  &  & (8,5,8)  & $82.6$ & $82.7$ & $68.6$ & $96.5$ & $96.8$ & $95.7$ & $95.7$ & $95.9$ & $92.6$ \\
&  &  & Adaptive  & $92.7$ & $92.7$ & $85.9$ & $\mathbf{99.4}$ & $\mathbf
{99.5}$ & $\mathbf{99.2}$ & $95.4$ & $95.1$ & $91.2$ \\ \hline
&  & (0.6,0.4) & Wilcoxon  & $99.3$ & $4.8$ & $4.8$ & $99.4$ & $45.9$ & $45.9$ & $95.5$ & $46.1$ & $46.0$ \\
&  &  & (8,5,8)  & $99.9$ & $13.1$ & $13.1$ & $99.8$ & $64.9$ & $64.9$ & $97.0$ & $60.7$ & $60.5$ \\
&  &  & Adaptive  & $\mathbf{100.0}$ & $30.8$ & $30.8$ & $\mathbf
{100.0}$ &  $\mathbf{86.5}$ & $\mathbf{86.5}%
$ & $98.4$ & $70.8$ & $70.6$ \\ \hline
500 & 500 & (0.5,0.0) & Wilcoxon  & $47.5$ & $$ & $$ & $96.3$ & $$ & $$ & $90.5$ & $$ & $$ \\
&  &  & (8,5,8)  & $68.4$ & $$ & $$ & $99.0$ & $$ & $$ & $90.7$ & $$ & $$ \\
&  &  & Adaptive  & $83.0$ & $$ & $$ & $\mathbf{99.8}%
$ & $$ & $$ & $87.8$ & $$ & $$ \\ \hline
&  & (0.5,0.5) & Wilcoxon  & $48.5$ & $48.8$ & $23.7$ & $91.1$ & $91.2$ & $88.3$ & $89.4$ & $89.5$ & $82.0$ \\
&  &  & (8,5,8)  & $68.7$ & $69.0$ & $47.7$ & $96.8$ & $96.7$ & $95.7$ & $90.3$ & $90.6$ & $82.9$ \\
&  &  & Adaptive  & $83.7$ & $83.9$ & $70.3$ & $\mathbf{99.5}$ & $\mathbf
{99.5}$ & $\mathbf{99.3}$ & $88.2$ & $88.7$ & $79.1$ \\ \hline
&  & (0.6,0.4) & Wilcoxon  & $97.8$ & $1.8$ & $1.7$ & $99.5$ & $45.9$ & $45.9$ & $94.7$ & $39.0$ & $38.6$ \\
&  &  & (8,5,8)  & $99.3$ & $5.8$ & $5.7$ & $99.7$ & $63.9$ & $63.9$ & $96.2$ & $49.5$ & $49.0$ \\
&  &  & Adaptive  & $99.9$ & $16.2$ & $16.2$ & $\mathbf{100.0}$ & $\mathbf
{85.7}$ & $\mathbf{85.7}$ & $96.9$ & $57.2$ & $56.2$ \\
\hline\end{tabular}
\end{table}%
%EndExpansion

\end{document}